\documentclass[acmsmall]{acmart}

\acmYear{2023}
\acmDOI{} 
\startPage{1}

\setcopyright{none}

\usepackage{wrapfig}

\usepackage{algorithm}
\usepackage{algpseudocode}
\usepackage{pgf}
\usepackage{tikz}
\usepackage{subcaption}

\usepackage{color}
\definecolor{GrayBgColor}{rgb}{0.9, 0.9, 0.9}

\newcommand*{\backref}[1]{}
\newcommand*{\backrefalt}[4]{%
    \ifcase #1%
          \or Cited on page~#2.%
          \else Cited on pages~#2.%
    \fi%
    }

\makeatletter
\newcommand{\xRightarrow}[2][]{\ext@arrow 0359\Rightarrowfill@{#1}{#2}}
\makeatother

\newcommand{\bi}{\begin{array}[t]{@{}l@{}}}
\newcommand{\ei}{\end{array}}
\newcommand{\ba}{\begin{array}}
\newcommand{\ea}{\end{array}}
\newcommand{\bda}[1]{\begin{displaymath}\ba{#1}}
\newcommand{\eda}{\ea\end{displaymath}}
\newcommand{\bp}{\begin{quote}\tt\begin{tabbing}}
\newcommand{\ep}{\end{tabbing}\end{quote}}

\newcommand{\ignore}[1]{}

\newcommand\Angle[1]{\langle#1\rangle}

\newcommand{\thread}[1]{\ensuremath{t_{#1}}}

\newcommand{\eventE}[1]{\ensuremath{e_{#1}}}

\newcommand{\HIGHLIGHT}[1]{\colorbox{GrayBgColor}{\ensuremath{#1}}}

\newcommand\evtAA{e}
\newcommand\evtBB{f}
\newcommand\evtCC{g}

\newcommand\evtSubject{e}
\newcommand\evtAcc{a}
\newcommand\evtAccA{\evtAcc'}
\newcommand\evtRel{r}
\newcommand\evtRelA{\evtRel'}

\newcommand{\proj}[2]{\textit{proj}(#1,#2)}

\newcommand{\threadID}[1]{\textit{thread}(#1)}

\newcommand{\pwrCand}{PWR}
\newcommand{\pwrCandGuarded}{PWR$^{cs}$}

\newcommand{\all}{{\mathit all}}

\newcommand{\AllLocksSym}{\LocksSym{\all}}
\newcommand{\acqVC}[1]{{\mathit Acq(#1)}}

\newcommand{\RW}[1]{\mathit{RC}(#1)}

\newcommand{\pwrVC}[1]{\mathit{PWR}(#1)}

\newcommand{\pp}{\ensuremath{\mathbin{\texttt{++}}}}

\newcommand{\preRead}[1]{\textit{read}(#1)}
\newcommand{\preWrite}[1]{\textit{write}(#1)}

\newcommand{\lockE}[1]{\textit{acq}(#1)} 
\newcommand{\unlockE}[1]{\textit{rel}(#1)} 
\newcommand{\readE}[1]{r(#1)}
\newcommand{\writeE}[1]{w(#1)}
\newcommand{\forkE}[1]{\textit{fork}(#1)}
\newcommand{\joinE}[1]{\textit{join}(#1)}

\newcommand{\ltTrace}[3]{#2 <_{\textit{tr}}^{#1} #3}
\newcommand{\ltTraceT}[4]{#2 <_{\textit{tr}}^{#1} #3 <_{tr}^{#1} #4}
\newcommand{\ltTr}[2]{#1 <_{\textit{tr}} #2}
\newcommand{\ltTrr}[3]{#1 <_{\textit{tr}} #2 <_{\textit{tr}} #3}
\newcommand{\ltTraceSym}[1]{<_{\textit{tr}}^{#1}}

\newcommand{\ltTO}[3]{#2 <_{to}^{#1} #3}

\newcommand{\ltWCP}[3]{#2 <_{wcp}^{#1} #3}
\newcommand{\ltWCPNoT}[2]{#1 <_{wcp} #2}
\newcommand{\ltWCPSym}[1]{<_{wcp}^{#1}}

\newcommand{\lteqWCP}[3]{#2 \leq_{wcp}^{#1} #3}

\newcommand{\ltSDPNoT}[2]{#1 <_{sdp} #2}

\newcommand{\DCTaskPar}{DC$_{tp}$}

\newcommand{\ltDCNoT}[2]{#1 <_{dc} #2}

\newcommand{\ltHB}[3]{#2 <_{hb}^{#1} #3}
\newcommand{\ltHBNoT}[2]{#1 <_{hb} #2}

\newcommand{\ltHBSym}[1]{<_{hb}^{#1}}

\newcommand{\ltHBCS}[3]{#2 <_{hb-cs}^{#1} #3}
\newcommand{\ltHBCSSym}[1]{<_{hb-cs}^{#1}}

\newcommand{\ltMustHB}[3]{#2 <_{mhb}^{#1} #3}
\newcommand{\ltMustHBB}[4]{#2 <_{mhb}^{#1} #3 <_{mhb}^{#1} #4}
\newcommand{\ltMHB}[2]{#1 <_{mhb} #2}
\newcommand{\ltMHBB}[3]{#1 <_{mhb} #2 <_{mhb} #3}

\newcommand{\ltMHBNoTSym}{<_{mhb}}

\newcommand{\concMustHB}[3]{#2 {\|}_{mhb}^{#1} #3}
\newcommand{\concPWR}[3]{#2 {\|}_{pwr}^{#1} #3}

\newcommand{\ACQ}[1]{\evtAcc_{#1}}
\newcommand{\THD}[1]{t_{#1}}

\newcommand{\nf}[1]{\mbox{\normalfont{#1}}}

\newcommand{\posP}[2]{\textit{pos}_{{\scriptstyle #1}}(#2)}

\newcommand{\events}[1]{\textit{events}(#1)}
\newcommand{\threads}[1]{\textit{threads}(#1)}
\newcommand\thd[1]{\ensuremath{\textit{thread} (#1)}}

\newcommand{\threadVC}[1]{\textit{Th}(#1)}

\newcommand\conflict[2]{\ensuremath{#1 \asymp #2}}

\newcommand{\FalseP}[2]{\textit{FP}^{#1} (#2)}
\newcommand{\FalseN}[2]{\textit{FN}^{#1} (#2)}

\newcommand{\FP}[1]{\textit{FP}(#1)}
\newcommand{\FN}[1]{\textit{FN}(#1)}

\newcommand\PredictDRP[2][{}]{\ensuremath{{P'}_{DR}^{#1} (#2)}}
\newcommand\PredictDR[2][{}]{\ensuremath{P_{DR}^{#1} (#2)}}
\newcommand\TrueDR[1][{}]{\ensuremath{G_{DR}^{#1}}}

\newcommand{\DataRace}[3]{#2 \bowtie^{#1} #3}

\newcommand\indexedcap{\mathbin{\cap'}}

\newcommand{\standard}{S}
\newcommand{\intrathread}{C}
\newcommand{\intrathreadthread}{CT}

\newcommand{\pwrsymbol}{\textit{pwr}}

\newcommand{\ltPWR}[3]{#2 <_{\pwrsymbol}^{#1} #3}

\newcommand{\ltPWRNoT}[2]{#1 <_{\pwrsymbol} #2}

\newcommand{\CSect}[4]{CS^{#2}_{#1}(#4)^{#3}} 
\newcommand{\CSec}[3]{CS_{#1}(#3)^{#2}}    
\newcommand{\CS}[2]{CS_{#1}(#2)}           

\newcommand{\csAR}[1]{CS_{C}(#1)}    

\newcommand{\StandardCS}[2]{{\mathit CS_{#1}}(#2)}

\newcommand{\AcqRelPair}[2]{\Angle{#1,#2}}

\newcommand{\StdCSect}[3]{\CSect{\standard}{#1}{#2}{#3}}
\newcommand{\StdCSec}[2]{\CSec{\standard}{#1}{#2}}

\newcommand{\IntraCSect}[3]{\CSect{\intrathread}{#1}{#2}{#3}}
\newcommand{\IntraCSec}[2]{\CSec{\intrathread}{#1}{#2}}

\newcommand{\LHeld}[3]{LH^{#2}_{#1}(#3)} 
\newcommand{\LH}[2]{LH_{#1}(#2)}           
\newcommand{\LHSym}[1]{{LH_{#1}}}           
\newcommand{\LocksSym}[1]{\ensuremath{\mathcal{L}_{#1}}}           

\newcommand{\IntraLHeld}[2]{\LHeld{\intrathread}{#1}{#2}}
\newcommand{\IntraLH}[1]{\LH{\intrathread}{#1}}

\newcommand{\CTTLHeld}[2]{\LHeld{\intrathreadthread}{#1}{#2}}
\newcommand{\CTTLH}[1]{\LH{\intrathreadthread}{#1}}

\newcommand{\StdLHeld}[2]{\LHeld{\standard}{#1}{#2}}
\newcommand{\StdLH}[1]{\LH{\standard}{#1}}

\newcommand{\StdLocksSym}{\LocksSym{\standard}}

\newcommand{\PWRLH}[1]{\LH{\pwrsymbol}{#1}}

\newcommand{\rwT}[1]{T^{\textit{rw}}_{#1}}
\newcommand{\rwTa}{\rwT{a}}

\newcommand{\rwTr}{T^{\textit{rw}}}

\makeatletter
\newdimen\legendxshift
\newdimen\legendyshift
\newcount\legendlines
\newcommand{\bclldist}{1mm}
\newcommand{\bclegend}[3][10mm]{%
	\legendxshift=0pt\relax
	\legendyshift=0pt\relax
	\xdef\legendnodes{}%
	\foreach \lcolor/\ltext [count=\ll from 1] in {#3}%
	{\global\legendlines\ll\pgftext{\setbox0\hbox{\bcfontstyle\ltext}\ifdim\wd0>\legendxshift\global\legendxshift\wd0\fi}}%
	\@tempdima#1\@tempdima0.5\@tempdima
	\pgftext{\bcfontstyle\global\legendxshift\dimexpr\bcwidth-\legendxshift-\bclldist-\@tempdima-0.72em}
	\legendyshift\dimexpr5mm+#2\relax
	\legendyshift\legendlines\legendyshift
	\global\legendyshift\dimexpr\bcpos-2.5mm+\bclldist+\legendyshift
	\begin{scope}[shift={(\legendxshift,\legendyshift)}]
		\coordinate (lp) at (0,0);
		\foreach \lcolor/\ltext [count=\ll from 1] in {#3}%
		{
			\node[anchor=north, minimum width=#1, minimum height=5mm,fill=\lcolor] (lb\ll) at (lp) {};
			\node[anchor=west] (l\ll) at (lb\ll.east) {\bcfontstyle\ltext};
			\coordinate (lp) at ($(lp)-(0,5mm+#2)$);
			\xdef\legendnodes{\legendnodes (lb\ll)(l\ll)}
		}
		\node[draw, inner sep=\bclldist,fit=\legendnodes] (frame) {};
	\end{scope}
}
\makeatother

\begin{document}

\title{Cross-thread critical sections and efficient dynamic race prediction methods}


\author{Martin Sulzmann}
\affiliation{
  \institution{Karlsruhe University of Applied Sciences}
  \streetaddress{Moltkestrasse 30}
  \city{Karlsruhe}
  \postcode{76133}
  \country{Germany}
}
\email{martin.sulzmann@gmail.com}

\author{Peter Thiemann}
\affiliation{
  \institution{University of Freiburg}
  \streetaddress{Georges-K{\"o}hler-Allee 079}
  \city{Freiburg}
  \postcode{79110}
  \country{Germany}
}
\email{thiemann@acm.org}




\begin{abstract}
  The lock set method and the partial order method are two main approaches to guarantee
  that dynamic data race prediction remains efficient.
  There are many variations of these ideas.
  Common to all of them is the assumption that the events
  in a critical section belong to the same thread.

  We have evidence that critical sections in the wild do
  extend across thread boundaries even if the surrounding acquire and
  release events occur in the same thread.
  We introduce the novel concept of a cross-thread critical section to
  capture such situations, offer a theoretical comprehensive framework, and study
  their impact on state-of-the-art data race analyses.

  For the sound partial order relation WCP we can show that the soundness claim
  also applies to cross-thread critical sections.
  For DCtp the occurrence of cross-thread critical sections invalidates the soundness claim.
  For complete partial order relations such as WDP and PWR,
  cross-thread critical sections help to eliminate more false positives.
  The same (positive) impact applies to the lock set construction.

  Our experimental evaluation confirms that cross-thread critical sections arise in practice.
  For the complete relation PWR, we are able to reduce the number of false positives.
  The  performance overhead incurred by tracking cross-thread critical
  sections slows down the analysis by 10\%-20\%, on average.
\end{abstract}

\maketitle

\section{Introduction}
\label{sec:introduction}

Concurrent programming techniques have become essential to fully
leverage the potential of modern multicore architectures. However,
concurrent programming is notoriously challenging and requires
programmers to exercise great care to avoid bugs. One significant
source of bugs in concurrent programs is data races, which occur when
multiple threads read and write shared data concurrently without proper
synchronization. Fortunately, there is a wealth of research and tools
available to assist in detecting and predicting the presence of data
races.

Dynamic analysis is one of the methods used for predicting data
races. This approach aims to anticipate a program's behavior by
examining a trace of events generated during the execution of a single
program run that terminates successfully. These events typically
involve acquiring or releasing locks or performing read and write
operations on global variables. By analyzing this program trace,
dynamic data race analysis can identify potential data races and help
programmers identify areas where synchronization mechanisms, such as
locks or atomic operations, should be employed to ensure thread safety
and prevent bugs.

The key difficulty lies in predicting whether the trace can be reordered
such that two conflicting memory operations appear adjacent to each
other, but without affecting causal relationships between the
events. This prediction should be accurate and efficient, where
accuracy refers to the analysis yielding a reasonable
number of false positives and/or false negatives. False positives occur
when the analysis incorrectly issues a data race warning, while false
negatives occur when the analysis fails to detect actual harmful
reorderings.

Efficiency pertains to the scalability of the
analysis to handle large program traces without significant performance
or resource degradation. Unfortunately,
the naive approach of exhaustive exploration to identify such
reorderings is doomed as the non-deterministic nature of
concurrent programs often allows for an exponentially large
number of reorderings.
Striking the right balance between accuracy
and efficiency is crucial in developing effective techniques for
predicting data races in concurrent programs.

There are two main approaches to retain efficiency.
\emph{Partial-order-based methods} derive a partial order from the
program trace, the earliest instance being Lamport's happens-before
relation~\cite{lamport1978time}.
This relation orders
a lock release before any acquisition of the same lock in the
rest of the trace.
Events belonging to the same thread are ordered according to the
program order reflected in the trace.
Two events are considered concurrent unless they are ordered by this relation.
If conflicting memory operations are deemed concurrent, the method
flags them as candidates for a data race.
\emph{Methods based on lock sets}~\cite{Dinning:1991:DAA:127695:122767}
infer for each event the set of locks held by the current thread at this
event.
Two events are considered concurrent unless they are protected by a common lock.
A data race warning is issued if the lock sets of conflicting memory operations are disjoint.
For both approaches there are numerous refinements of the original ideas.

The standard lock set method has no false negatives, but it may exhibit false positives.
To reduce the number of false positives one popular approach is to
combine it with some form of partial order relation. A data race warning is only issued
if the lock sets are disjoint and conflicting events are unordered.
For example, the data race detection tools Eraser~\cite{Savage:1997:EDD:265924.265927}
and TSan~V1~\cite{serebryany2009threadsanitizer} include fork/join dependencies
that must be respected under any trace reordering.
The partial orders WDP~\cite{10.1145/3360605} and PWR~\cite{10.1145/3360605,10.1145/3426182.3426185}
additionally include some dependencies among critical sections.
The goal of WDP and PWR is to strengthen must-dependencies such as fork/join
without ruling out viable trace reorderings.
Analyses based on WDP or PWR remain complete (no false negatives).

In contrast, partial order-based methods can lead to false negatives. For example, HB is sensitive to the order of lock operations executed by the program.
It is unable to predict races that result from schedules that reorder lock acquisition.
CP~\cite{Smaragdakis:2012:SPR:2103621.2103702} and its improvement
WCP~\cite{Kini:2017:DRP:3140587.3062374} are refinements of
HB that do consider schedules that acquire locks in a different order.
Thus, some races, but not necessarily all, that result from alternative schedules can be found.
SDP~\cite{10.1145/3360605} is a further relaxation of WCP to detect even more races.
CP, WCP, and SDP weaken the happens-before relation as much as possible
without admitting invalid trace reorderings.
DC~\cite{Roemer:2018:HUS:3296979.3192385} is a partial order that is
weaker than WCP, but stronger than WDP.
It requires a subsequent vindication phase to reduce the number of
false positives by attempting to construct a reordered trace
that exhibits the race. {\DCTaskPar} \cite{8894270} is a seemingly
simpler variant which is specialized for single locks.

There are further approaches to dynamic data race analysis that
have scalability issues for large programs. We discuss them in the
related work (Section~\ref{sec:related}).

All of the aforementioned works rely on the same traditional concept
of a critical section. By default, a critical section refers to a
\emph{sequence of events within a single thread}, bracketed by an
acquire and the next following release operation on the same lock.
We generalize critical sections to sequences of \emph{events that must happen between an
acquire and its corresponding release operation}. Crucially, our
notion of a \emph{cross-thread critical section} may also include events
from threads other than the acquiring thread, but we still assume that the acquire and its
corresponding release event come happen in the same thread!
These situations arise in practice and affect the accuracy of data
race prediction methods, thereby impacting the outcome of the analysis.

This work makes the following contributions:

\begin{itemize}
\item We propose the novel concept of a cross-thread critical section (CTCS)
      and introduce the associated lock set construction (Section~\ref{sec:cross-thread}).
\item We study the theoretical impact of CTCS critical sections
      on lock set-based dynamic data race prediction methods
      (section~\ref{sec:lockset-race-pred}). We show that they never
      increase the number of false positives. False negatives may
      arise, but we give a new thread-indexed CTCS construction that retains the
      no-false-negatives property of the standard lock set-based approach.
\item We study the impact of CTCS on several popular, scalable partial
  order-based dynamic data race prediction methods.
  We find that the pure HB method is not affected by CTCS; the soundness proof of WCP
  can be adapted to CTCS; WDP and
  PWR can be adapted while retaining completeness;
  soundness of \DCTaskPar, DC in the context of task parallel programs, is compromised and it is not
  clear to the authors how to fix it.
\item Section~\ref{sec:experiments} gives empirical evidence that CTCS
  occur in practice in a standard suite of traces and
  decrease the number of false positives in some instances.
  Our implementation relies on PWR.
  The adaptation to CTCS results in a slowdown of 20\% on the average (worst case slowdown is 2.7x).

\end{itemize}

Section~\ref{sec:overview} presents an informal, example-driven overview of the
contributions of the paper. Section~\ref{sec:prelim} establishes
mathematical preliminaries and notation.



\section{Overview}
\label{sec:overview}

This section offers an informal overview of the impact of cross-thread critical
sections on dynamic data race prediction methods considered in the
literature. As mentioned in the introduction, we concentrate on
scalable methods, hence
we review lock set-based methods, offer a first definition of
cross-thread critical sections, and then discuss a selection of
partial order-based methods for dynamic race prediction.

\begin{figure}
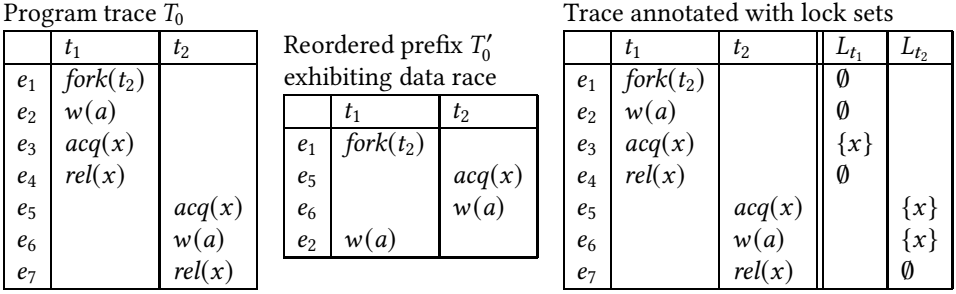

\bda{lcl}

\ba{@{}l@{}}

\mbox{Program trace $T_0$}

\\

\ba{|l|l|l|}
\hline  & \thread{1} & \thread{2} \\ \hline
\eventE{1}  & \forkE{\thread{2}}&\\
\eventE{2}  & \writeE{a}&\\
\eventE{3}  & \lockE{x}&\\
\eventE{4}  & \unlockE{x}&\\
\eventE{5}  & &\lockE{x}\\
\eventE{6}  & &\writeE{a}\\
\eventE{7}  & &\unlockE{x}\\

\hline \ea{}
\ea

 &

\ba{@{}l@{}}

\mbox{Reordered prefix $T_0'$}
\\ \mbox{exhibiting data race}

\\

\ba{|l|l|l|}
\hline  & \thread{1} & \thread{2}\\ \hline
\eventE{1}  & \forkE{\thread{2}}&\\
\eventE{5}  & &\lockE{x}\\
\eventE{6}  & & \writeE{a}\\
\eventE{2}  & \writeE{a}&\\
 \hline \ea{}

 \ea

 &

 \ba{@{}l@{}}

 \mbox{Trace annotated with lock sets}

\\

 \ba{|l|l|l||l|l|}
\hline  & \thread{1} & \thread{2} & L_{t_1} & L_{t_2} \\ \hline
\eventE{1}  & \forkE{\thread{2}}&& \emptyset &  \\
\eventE{2}  & \writeE{a}&& \emptyset &  \\
\eventE{3}  & \lockE{x}&& \{x\} &  \\
\eventE{4}  & \unlockE{x}&& \emptyset &  \\
\eventE{5}  & &\lockE{x} &  & \{ x \} \\
\eventE{6}  & &\writeE{a} &  & \{ x \} \\
\eventE{7}  & &\unlockE{x} &  & \emptyset \\

\hline \ea{}

\ea

 \eda
 \caption{Data race prediction using lock sets}
\label{fig:ex14}
\end{figure}

\subsection{Lock sets for data race prediction}

We review the use of lock sets for data race prediction with the example in figure~\ref{fig:ex14}.
The diagram on the left represents a program run by a trace of events $T_0$.
It visualizes the interleaved execution of the program using
a tabular notation with a separate column for each thread and one
event per row.
The textual order (from top to bottom) reflects the observed temporal order of events.

Each event takes place in a specific thread and represents an operation (formally defined in section~\ref{sec:prelim}).
Operation $\forkE{t}$ starts a new thread with ID $t$ and
operation $\joinE{t}$ synchronizes with the termination of thread $t$.
We use $x, y, z$ for locks and  $a,b,c$ for shared variables.
Operations $\lockE{x}/\unlockE{x}$ acquire/release lock $x$.
Operations $\readE{a}/\writeE{a}$ are shared memory read and write operations on $a$.
The same operation may appear multiple times in a trace, thus
we use indices as in $e_i$ to uniquely identify events in the trace.

Trace $T_0$ contains two conflicting memory operations in events $e_2$ and $e_6$.
Event $e_6$ is protected by lock $x$ whereas event $e_2$ is unprotected.
The conflict in $T_0$ appears to be harmless because
all operations in thread $t_1$ are completed before the operations in thread $t_2$ are executed.
However, there is a reordered prefix of $T_0$ (shown in $T_0'$) that puts the two write
operations back to back, which indicates a data race.

Instead of considering all possible reorderings, the lock set method computes
for each event the set of locks held when processing this event.
Lock sets $L_t$ are computed separately for each thread $t$.
Initially, all lock sets are empty.
Processing operation $\lockE{x}$ in thread $t$ adds $x$ to $L_t$.
Processing operation $\unlockE{x}$ in thread $t$ removes $x$ from $L_t$.

The right diagram in figure~\ref{fig:ex14} shows trace $T_0$ annotated
with lock sets. Columns $L_{t_1}$ and $L_{t_2}$ track the respective
lock set after the operation in the trace.
The observation that the lock sets of the conflicting events $e_2$ and
$e_6$ are disjoint results in a data race warning.

\begin{figure}

\bda{l}

 \mbox{Trace annotated with (standard) lock sets $L$ and (cross-thread) lock sets $C$}

\\

\ba{|l|l|l|l||l|l|l||l|l|l|}
\hline  & \thread{1} & \thread{2} & \thread{3} & L_{t_1} & L_{t_2} & L_{t_3} & C_{t_1} & C_{t_2} & C_{t_3} \\ \hline
\eventE{1}  & \forkE{\thread{3}}&& & \emptyset &&& \emptyset &&  \\
\eventE{2}  & \lockE{x}&&   & \{ x \} &&& \{x  \} && \\
\eventE{3}  & \forkE{\thread{2}}&&  & \{ x \} &&& \{ x  \} && \\
\eventE{4}  & &\writeE{a}&  &  & \HIGHLIGHT{\emptyset} && & \HIGHLIGHT{\{ x \}} & \\
\eventE{5}  & \joinE{\thread{2}}&& & \{ x \} &&& \{ x  \} && \\
\eventE{6}  & \unlockE{x}&& & \emptyset &&& \emptyset && \\
\eventE{7}  & &&\lockE{x} & & & \{ x \} && & \{ x \} \\
\eventE{8}  & &&\writeE{a} & & & \{ x \} && & \{ x \} \\
\eventE{9}  & &&\unlockE{x} & & & \emptyset && & \emptyset \\

\hline \ea{}

\eda

 \caption{Cross-thread critical sections}
\label{fig:exCS_2}
\end{figure}

\subsection{Cross-thread critical sections}

The use of locks allows us to regard certain subtraces as critical sections.
We identify a critical section by an acquire event $a$ and its
matching release event $r$.
Events $e$ between $a$ and $r$ are part of this critical section, written $e \in CS(a,r)$.
The standard assumption is that \emph{$e$, $a$, and $r$ belong to the same thread} and $e$ appears between $a$ and $r$
in the textual order of the trace.
This assumption goes back to the original intuition of a critical
section as a sequence of instructions that is executed
atomically \cite{DBLP:journals/cacm/Dijkstra65}.
In the present context, we need a more general definition.

Consider the sequence of events  $[e_2, \dots, e_6]$ in figure~\ref{fig:exCS_2}.
We regard this sequence as a critical section
\emph{that extends across multiple threads}.
We call such critical sections \emph{cross-thread} critical sections, written $\csAR{\cdot}$, characterized as follows:

\begin{enumerate}
\item $r$ is the matching release for $a$ with no other release on the same lock in between.

\item $a$ and $r$ belong to the same thread.

\item $e \in \csAR{a,r}$ if $a \ltMHBNoTSym e \ltMHBNoTSym r$
  where $\ltMHBNoTSym$  denotes the \emph{must} happen-before relation.
\end{enumerate}
The novelty lies in the third condition.
Event $e$ can be in any thread as long as the acquire $a$ must
happen-before $e$, which in turn
must happen-before the release $r$.
That is,  $e$ is surrounded by $a$ and $r$ in any valid reordering of
the trace.
In our example, this ordering is guaranteed by a fork-join dependency, but any
other form of happens-before dependency (like write-read) would also work.
Formal definitions of (cross-thread) critical sections are given in section~\ref{sec:cross-thread}.

Cross-thread critical sections are a significant addition to the
toolbox of lock set-based data race prediction.
Returning to the example in figure~\ref{fig:exCS_2}, we can see that
the lock set of $e_4$ is empty (column $L_{t_3}$) if we rely on the standard lock set construction.
This outcome leads to a data race warning between $e_4$ and $e_8$ that
cannot be materialized by a reordering, i.e.,  a false positive.
However, with lock sets based on cross-thread critical sections, the
lock set of $e_4$ is $\{x\}$ (column $C_{t_3}$), which  eliminates the false positive.

In section~\ref{sec:lockset-race-pred} we make this claim formal by
showing that a lock set construction based on
cross-thread critical sections strictly improves over the standard construction:
it exhibits fewer false positives than the standard construction
without introducing false negatives.

\subsection{Partial order methods for data race prediction}

Partial order methods define an ordering, say $P$, on events in a
trace and consider unordered events as concurrent.
If conflicting events are unordered they are potentially in a race.
Ideally, $P$ would relate events $e$ and $f$ iff $e$ happens before
$f$ in any valid reordering of the trace. In this case, $P$ would be
\emph{sound} (no false positives) and \emph{complete} (no false
negatives).
Practical methods give up on soundness or completeness to obtain an
efficiently computable relation.

Relations like HB, WCP, and SDP focus on soundness and thus \emph{overapproximate} the must happen-before relation.
That is, if $e$ must happen before $f$, then $e$ is before $f$ in,
say, HB.
Complete relations like WDP and PWR \emph{underapproximate} the must happen-before relation.
That is, if neither $e$ must happen before $f$ nor $f$ before $e$,
then $e$ and $f$ are not ordered by, say, WDP.
The relation DC is a bit of an outlier as it is neither sound nor complete.

In the following, we discuss these partial order relations and study the impact of cross-thread
critical sections on each of them.

\mbox{}
\\
\noindent
{\bf Happens-before.}
Lamport's happens-before (HB) relation \cite{lamport1978time}
prescribes the trace order for all events in the same thread and that
an acquire happens after any release of the same lock that occurs
textually before it in the trace:
\begin{description}
\item[\nf{(PO)}]  $\ltHBNoT{e}{f}$ if $e$ and $f$ are in the same
  thread and $e$ appears before $f$ in the trace.
\item[\nf{(AcqRel)}] $\ltHBNoT{r}{a}$ if $a$ is an acquire and $r$ a release event on the same lock
  and $r$ appears before $a$ in the trace.
\end{description}
As an example consider the left trace in figure~\ref{fig:wcp}.
HB dictates that $\ltHBNoT{e_4}{e_5}$. Therefore the two conflicting
events $e_3$ and $e_7$ are HB-ordered and HB reports no race.

Thanks to rule \nf{(AcqRel)}, HB maintains the order of critical sections according to the
trace. Therefore HB already takes care of cross-thread critical
sections.
In the example in figure~\ref{fig:exCS_2}, rule \nf{(AcqRel)} enforces $\ltHBNoT{e_6}{e_7}$.
Due to the join dependency we also have $\ltHBNoT{e_4}{e_5}$.
By \nf{(PO)} we conclude that $\ltHBNoT{e_4}{e_8}$.
Hence, cross-thread critical sections neither affect the soundness of
HB nor lead to extra false negatives.

The situation is different for approaches that may reorder critical sections.

\begin{figure}

\bda{lcl}

\ba{|l|l|l|}
\hline  & \thread{1} & \thread{2}\\ \hline
\eventE{1}  & \forkE{\thread{2}}&\\
\eventE{2}  & \lockE{x}&\\
\eventE{3}  & \writeE{a}&\\
\eventE{4}  & \unlockE{x}&\\
\eventE{5}  & &\lockE{x}\\
\eventE{6}  & &\unlockE{x}\\
\eventE{7}  & &\writeE{a}\\

\hline \ea{}

& &

\ba{l}
\mbox{Reordering exhibiting race}

\\

\ba{|l|l|l|}
\hline  & \thread{1} & \thread{2}\\ \hline
\eventE{1}  & \forkE{\thread{2}}&\\
\eventE{5}  & &\lockE{x}\\
\eventE{6}  & &\unlockE{x}\\
\eventE{2}  & \lockE{x}&\\
\eventE{7}  & &\writeE{a}\\
\eventE{3}  & \writeE{a}&\\

\hline \ea{}

\ea

\eda

  \caption{Comparison between HB and WCP. HB: no race (false
    negative). WCP: race}
  \label{fig:wcp}
\end{figure}

\mbox{}
\\
\noindent
{\bf Weak-causally precedes.}
The WCP relation (weak-causally precedes)~\citep{Kini:2017:DRP:3140587.3062374} only orders
critical sections if they contain conflicting events:
\begin{description}
\item[\nf{(a)}] $\ltWCPNoT{r_1}{f}$ if $e \in CS(a_1,r_1)$, $f \in CS(a_2,r_2)$,
          $CS(a_1,r_1)$ appears before $CS(a_2,r_2)$ in the trace and $e$ and $f$ are in a conflict.
\end{description}
There are no conflicts between critical sections in
Figure~\ref{fig:wcp}.  The two write events $e_3$ and $e_7$ are not WCP-ordered and
WCP correctly finds a race, as materialized by the reordering on the right.

WCP checks for conflicts using the standard notion of a critical section $CS(\cdot)$.
This sounds like trouble in situations as shown in figure~\ref{fig:exCS_2}.
Events $e_4$ and $e_8$ are in a conflict, but $e_4$ is not part of a
critical section (from WCP's point of view) so that rule \nf{(a)} is
not sufficient to order $e_4$ before $e_8$.

However, WCP imposes two additional rules:
\begin{description}
\item[\nf{(b)}] Fork-join dependencies are WCP-ordered.
\item[\nf{(c)}] WCP composes to the left and right with the HB relation.
  That is, $\ltWCPNoT{e}{f}$ if either $\ltWCPNoT{e}{g}$ and $\ltHBNoT{g}{f}$ or
                       $\ltHBNoT{e}{g}$ and $\ltWCPNoT{g}{f}$.
 \end{description}
From (b) we obtain $\ltWCPNoT{e_4}{e_5}$.
The happens-before relation yields $\ltHBNoT{e_5}{e_8}$.
Thus, we derive $\ltWCPNoT{e_4}{e_8}$ via rule (c).
We observe that WCP is unaware of cross-thread critical sections when applying rule (a).
Thanks to the additional rules (b) and (c), WCP will not falsely issue
a race warning, so its soundness is not affected for this example.

\begin{figure}

\bda{lcl}

\ba{|l|l|l|}
\hline  & \thread{1} & \thread{2}\\ \hline
\eventE{1}  & \forkE{\thread{2}}&\\
\eventE{2}  & \writeE{a}&\\
\eventE{3}  & \lockE{y}&\\
\eventE{4}  & \writeE{b}&\\
\eventE{5}  & \unlockE{y}&\\
\eventE{6}  & &\lockE{y}\\
\eventE{7}  & &\writeE{b}\\
\eventE{8}  & &\unlockE{y}\\
\eventE{9}  & &\writeE{a}\\

 \hline \ea{}

& &

\ba{l}
\mbox{Reordering exhibiting race}

\\

\ba{|l|l|l|}
\hline  & \thread{1} & \thread{2}\\ \hline
\eventE{1}  & \forkE{\thread{2}}&\\
\eventE{6}  & &\lockE{y}\\
\eventE{7}  & &\writeE{b}\\
\eventE{8}  & &\unlockE{y}\\
\eventE{9}  & &\writeE{a}\\
\eventE{2}  & \writeE{a}&\\

 \hline \ea{}

\ea

\eda

  \caption{WCP no race (false negative), SDP race}
  \label{fig:exWCP_7b}
\end{figure}

\mbox{}
\\
\noindent
{\bf Strong dependently precedes.}
The SDP relation (strong dependently precedes)~\cite{10.1145/3360605}
relaxes WCP with the goal to eliminate some false negatives. To this
end, SDP modifies WCP's rule \nf{(a)} to ignore conflicts between
write operations:
\begin{description}
\item[\nf{(a')}] $\ltSDPNoT{r_1}{f}$ if $e \in CS(a_1,r_1)$, $f \in CS(a_2,r_2)$,
  $CS(a_1,r_1)$ appears before $CS(a_2,r_2)$ in the trace and $e$ and $f$ are in a conflict and either $e$ or $f$ is a read
   operation.
\end{description}

The example in figure~\ref{fig:exWCP_7b} illustrates the difference.
For WCP, we find $\ltWCPNoT{e_6}{e_8}$ via rule (a)
and by applying rule (c) we obtain that $\ltWCPNoT{e_3}{e_{10}}$.
So, WCP is unable to detect the race between conflicting events $e_3$ and $e_{10}$.
In contrast, SDP replaces rule \nf{(a)} by rule \nf{(a')} with the
result that $e_3$ and $e_{10}$ are unordered under SDP and the race is
detected.

Moreover,  SDP does not order the conflicting events $e_5$ and $e_8$.
This conflict would be a false positive, but SDP does not report it
because the lock sets of $e_5$ and $e_8$ are not disjoint.

\begin{figure}

\bda{lcl}

\ba{|l|l|l|l|}
\hline  & \thread{1} & \thread{2} & \thread{3}\\ \hline
\eventE{1}  & \forkE{\thread{3}}&&\\
\eventE{2}  & \lockE{x}&&\\
\eventE{3}  & \forkE{\thread{2}}&&\\
\eventE{4}  & &\lockE{y}&\\
\eventE{5}  & &\unlockE{y}&\\
\eventE{6}  & \joinE{\thread{2}}&&\\
\eventE{7}  & \writeE{a}&&\\
\eventE{8}  & \unlockE{x}&&\\
\eventE{9}  & &&\lockE{y}\\
\eventE{10}  & &&\lockE{x}\\
\eventE{11}  & &&\unlockE{x}\\
\eventE{12}  & &&\writeE{a}\\
\eventE{13}  & &&\unlockE{y}\\

 \hline \ea{}

& &

\ba{l}

\mbox{Reordering that gets stuck}

\\

\ba{|l|l|l|l|}
\hline  & \thread{1} & \thread{2} & \thread{3}\\ \hline
\eventE{1}  & \forkE{\thread{3}}&&\\
\eventE{2}  & \lockE{x}&&\\
\eventE{3}  & \forkE{\thread{2}}&&\\
\eventE{9}  & &&\lockE{y}\\
\eventE{10}  & && \HIGHLIGHT{\lockE{x}}\\

\eventE{4}  & &\HIGHLIGHT{\lockE{y}}&\\

 \hline \ea{}

\ea

\eda

 \caption{Trace with no predictable race, nor a predictable deadlock}
\label{fig:exWCP_9d}
\end{figure}

\mbox{}
\\
\noindent
{\bf WCP soundness in the presence of cross-thread critical sections.}
WCP is (weakly) sound in the following way~\cite{Kini:2017:DRP:3140587.3062374,10.1145/3360605}:
If some trace exhibits a WCP-race then either there is a predictable race or a predictable deadlock.
This property applies to the first race reported.
Predictable race means that we can reorder the trace such that the two conflicting events
appear right next two each other.
Predictable deadlock means that a set of threads is blocked because
each thread fails to acquire a lock because this lock has been acquired
by some of the other threads.

The WCP relation is based on the standard definition of critical sections.
Hence, we ask the question if cross-thread critical sections
threaten soundness of WCP.
Consider the trace in figure~\ref{fig:exWCP_9d}.

The conflicting events are not in a predictable race as
there is no reordering under which we can place them next to each
other, but there is a reordering that gets stuck (shown on the right
of figure~\ref{fig:exWCP_9d}).
In this reordering,
thread $t_2$ attempts to acquire lock $y$, which is held by $t_3$ and
thread $t_3$ attempts to acquire lock $x$, which is held by $t_1$.
See the highlighted blocked operations.

However, this form of \emph{stuckness} is not a predictable deadlock.
The difference to a (standard) predictable deadlock situation is that further threads, beyond the deadlocked threads $t_2$ and $t_3$, are involved.
Lock $x$ is held by thread $t_1$, not by thread $t_2$, due to the cross-thread critical section.

We conclude that the standard notion of a predictable deadlock is insufficient in general
to capture all stuck situations that arise due to cycles among acquire operations.
Fortunately, WCP imposes strong conditions so that the conflicting events
$e_7$ and $e_{12}$ are ordered under WCP.

Rule (a) does not apply because conflicting events $e_7$ and $e_{12}$ are not part of critical sections
that share the same lock.
From rule (b) we obtain $\ltWCPNoT{e_5}{e_6}$.
In combination with rule (c) we find that $\ltWCPNoT{e_2}{e_{11}}$.
At this point, it seems that $e_7$ and $e_{12}$ are not ordered under WCP.
However, WCP imposes the following fourth rule.
\begin{description}
\item[\nf{(d)}] $\ltWCPNoT{r_1}{r_2}$ if $\ltWCPNoT{a_1}{r_2}$ for critical sections
     $CS(a_1,r_1)$ and $CS(a_2,r_2)$.
\end{description}
Thus, we can conclude that $\ltWCPNoT{e_8}{e_{11}}$.
In combination with rule (c) we then obtain that $\ltWCPNoT{e_7}{e_{12}}$.

In appendix~\ref{sec:wcp}, we examine the WCP soundness proof in detail.
The proof only assumes standard critical sections.
We show which parts of the proofs are affected and show how the arguments can be generalized
to take into account cross-thread critical sections.
The same observations should apply to SDP as well but we have not yet fully worked out all details.

\begin{figure}

\bda{l}

\ba{|l|l|l|l|l|}
\hline  & \thread{1} & \thread{2} & \thread{3} & \thread{4}\\ \hline
\eventE{1}  & \forkE{\thread{3}}&&&\\
\eventE{2}  & \lockE{x}&&&\\
\eventE{3}  & \writeE{a}&&&\\
\eventE{4}  & \forkE{\thread{2}}&&&\\
\eventE{5}  & &\writeE{b}&&\\
\eventE{6}  & \joinE{\thread{2}}&&&\\
\eventE{7}  & \unlockE{x}&&&\\
\eventE{8}  & &&\lockE{x}&\\
\eventE{9}  & &&\forkE{\thread{4}}&\\
\eventE{10}  & &&&\writeE{b}\\
\eventE{11}  & &&\readE{a}&\\
\eventE{12}  & &&\unlockE{x}&\\

\hline \ea{}

\eda

 \caption{\DCTaskPar\ race (false positive), loosely released cross-thread critical section}
\label{fig:exDC_2b}

\end{figure}

\mbox{}
\\
\noindent
{\bf Does-not-commute for task parallel programs.}
The does-not-commute (DC)
relation~\cite{Roemer:2018:HUS:3296979.3192385} omits rule \nf{(c)}
from WCP, i.e., DC does not compose with HB.
This modification leads to false positives, even for the first race reported.
\citet{8894270} consider the variant \DCTaskPar\ in the context of task parallel programs.
\DCTaskPar\ applies the same partial order rules as DC but assumes that there is a single lock only.
Hence, a successful program run will not end in a deadlock.
We might expect that it becomes easier to obtain a soundness result
as we do not need to consider the case of a predictable deadlock.
This is not the case.

For the example in figure~\ref{fig:exCS_2}, \DCTaskPar\ wrongly claims that there is a race.
The two conflicting events $e_4$ and $e_8$ are unordered under \DCTaskPar, because
(a) \DCTaskPar\ uses the standard notion of a critical section (like WCP and SDP), and
(b) \DCTaskPar\ does not compose with HB (unlike WCP and SDP).
We conclude that \DCTaskPar\ is unsound in the presence of cross-thread critical sections.

We might hope to restore soundness for \DCTaskPar\ by adjusting rule (a) to replace standard critical sections $\CS\cdot$
with cross-thread critical sections $\csAR{\cdot}$.
\begin{description}
\item[\nf{(a'')}] $\ltDCNoT{r_1}{f}$ if $e \in \csAR{a_1,r_1}$, $f \in \csAR{a_2,r_2}$,
          $\csAR{a_1,r_1}$ appears before $\csAR{a_2,r_2}$ in the trace and $e$ and $f$ are in a conflict.
\end{description}
This change eliminates the false positive in figure~\ref{fig:exCS_2},
but is insufficient in general, as demonstrated with the trace in figure~\ref{fig:exDC_2b}.
The conflicting events $e_5$ and $e_{10}$ are in a \DCTaskPar\ race,
but this race  is not predictable.
Event $e_{10}$ is covered by a ``loosely released'' critical section.
In combination with the write-read dependency among $e_3$ and $e_{11}$,
we conclude that $e_7$ must happen-before $e_8$.

In summary, we see that the adjustments to obtain soundness for \DCTaskPar\ are non-trivial
as we also need to reason about further variants of cross-thread critical sections.
Hence, we focus our attention on partial order methods that are generally unsound (considering the first race reported)
with the goal to eliminate as many false positives as possible.

\mbox{}
\\
\noindent
{\bf WDP and PWR.}
The relations WDP~\cite{10.1145/3360605}
and its improvement PWR~\cite{10.1145/3426182.3426185}
underapproximate must happen-before relations
by weakening rule (a):
\begin{description}
\item[\nf{(a''')}] $\ltPWRNoT{r_1}{f}$ if $e \in CS(a_1,r_1)$, $f \in CS(a_2,r_2)$,
  $CS(a_1,r_1)$ appears before $CS(a_2,r_2)$ in the trace and $e$ is a write and $f$ is a read operation.
\end{description}
If $\ltPWRNoT{e}{f}$, then $e$ appears before $f$ in any correct reordering.
Hence, a partial order-based race check using WDP and PWR guarantees that there are no false negatives.
Like SDP, they make use of the lock set to eliminate some, but not all false positives.


WDP and PWR remain complete (no false negatives) in the presence of
cross-thread critical sections, but
we can improve their precision (eliminate more false positives ).
For example, in figure~\ref{fig:exCS_2} the
conflicting events $e_4$ and $e_8$ are unordered under PWR and their standard lock set is disjoint.
By using cross-thread lock sets we can eliminate false positives like
this one.
We can also incorporate cross-thread critical sections into rule (a''') to eliminate
further false positives.
Details are discussed in section~\ref{sec:pwr}.

\subsection{Implementation and experiments}

For our experimental evaluation we enhanced the PWR data race predictor
so that PWR is aware of cross-thread critical sections. The goal was to answer the
following research questions.

\begin{description}
\item[RQ1] Can we compute cross-thread critical sections and the associated
  lock sets efficiently?

\item[RQ2] What is the effect on the analysis results of  incorporating
  cross-thread critical section?
\end{description}

As PWR underapproximates the must happen-before relation,
we can use PWR to underapproximate the cross-thread lock set.
That is, a PWR-computed cross-thread lock set may be smaller than the
``true'' cross-thread lock set,
but the computation can be done efficiently and our experiments
confirm that there are a number of examples where we can eliminate
false positives.
Details are discussed in section~\ref{sec:experiments}.

\section{Preliminaries}
\label{sec:prelim}

\noindent
{\bf Events and Traces.}
We consider concurrent programs with shared variables and locks.
Concurrency primitives are \textit{acq}uire and \textit{rel}ease of locks (mutexes) as well
as \textit{fork} to start a new thread and \textit{join} to synchronize with its termination.

\begin{definition}[Events and Traces]
\label{def:run-time-traces-events}
\bda{lcll}
T & ::= & [] \mid e : T   & \mbox{Traces}
 \\ e & ::= & (\alpha, t,op) & \mbox{Events}
  \\ op & ::= &  \readE{a}
           \mid \writeE{a}
           \mid \lockE{x}
           \mid \unlockE{x}
           \mid \forkE{t}
           \mid \joinE{t}
           & \mbox{Operations}
\\ t,s,u & ::= & t_1 \mid t_2 \mid ... & \mbox{Thread ids}
\\  x, y, z & ::= & x_1\mid ... & \mbox{Lock variables}
\\  a, b, c & ::= & a_1 \mid ... & \mbox{Shared variables}
\\ \alpha, \beta,\delta & ::= & 1 \mid 2 \mid ... & \mbox{Unique event identifiers}
\eda
\end{definition}
A trace $T$ is a list of events reflecting a single execution of a
concurrent program under the sequential consistency memory
model~\cite{Adve:1996:SMC:619013.620590}. We write $[o_1,\dots,o_n]$
for a list of objects as a shorthand of $o_1:\dots:o_n:[]$ and use the
operator  $\pp$ for list concatenation.

An event $e$ is represented by a triple $(\alpha, t,op)$
where
$\alpha$ is a unique event identifier,
$op$ is an operation, and
$t$ is the thread id in which the operation took place.
The main thread has thread id $\thread{1}$.
The unique event identifier allows us to unambiguously identify events in case of trace reordering.

The operations $\readE{a}$ and $\writeE{a}$ denote read and write
on a shared variable $a$.
We let $\lockE{x}$ and $\unlockE{x}$ denote acquire and release of a lock $x$.
We write $\forkE{t}$ for the creation of a new thread with thread id $t$.
We write $\joinE{t}$ for a join with a thread with thread id $t$.

Our tabular notation for traces has one column per thread. The events
for a thread are lined up in the thread's column and the trace
position corresponds to the row number.

We write $\evtAA = (t,op)$ as a shorthand for $\evtAA = (\alpha, t,
op)$ and $\thd{\evtAA} = t$ to extract the thread id from this event.
The notation $e \in T$ indicates that $T = [e_1, \dots, e_n]$ and $e = e_k$, for some
$1\le k\le n$. In this case, we define $\posP{T}{e} = k$.
The set of events in a trace is $\events{T} = \{
e \mid e \in T\}$.
The set of thread ids in a trace is $\threads{T} = \{ \thd{e} \mid e \in T \}$.


For trace $T$ and events $e, f \in \events{T}$,
we define $\ltTrace{T}{e}{f}$ if $\posP{T}{e} < \posP{T}{f}$.
We write $\ltTr{e}{f}$ if the context uniquely identifies the trace $T$.


\noindent
{\bf Well-formedness.}
Traces must be well-formed.
We adopt the  sequential consistency conditions for concurrent objects
of Huang and others~\cite{Huang:2014:MSP:2666356.2594315}.
For example, events can only happen in a thread between its creation
and its termination, each release must be preceded an acquire on the same lock etc.

\begin{definition}
\label{def:well-formed-trace}
A trace $T$ is \emph{well-formed} if the following conditions are satisfied:
\begin{description}
\item[Lock-1:]
  For each pair of acquire events
  $\evtAcc = (t,\lockE{y}), \evtAccA = (s,\lockE{y}) \in T$
  where  $\ltTr{\evtAcc}{\evtAccA}$
  there exists a release event $\evtRel = (t,\unlockE{y}) \in T$ such that
  $\ltTrr{\evtAcc}{\evtRel}{\evtAccA}$.

 \item[Lock-2:]
   For each release event $\evtRel = (t,\unlockE{y}) \in T$
   there exists an acquire $\evtAcc = (t,\lockE{y}) \in T$
   such that
   $\ltTr{a}{r}$
   and there is no release event $\evtRelA = (s,\unlockE{y}) \in T$
   with $\ltTrr{\evtAcc}{\evtRelA}\evtRel$.

 \item[Fork-1:]
   For each thread id $t \ne t_1$ there exists at most one event with operation $\forkE{t}$
   in trace~$T$ and $\forkE{t_1}$ does not appear at all.

 \item[Fork-2:]
   For each event $\evtAA = (t,op) \in T$ where $t \not= t_1$ there exists
   $\evtBB = (s,\forkE{t}) \in T$ where
   $\ltTr{\evtBB}{\evtAA}$.

 \item[Join:]
   For each join event $\evtAA = (s,\joinE{t}) \in T$ we have
   that $s \ne t$ and for all events $\evtBB = (t,op) \in T$  we find that
   $\ltTr{\evtBB}{\evtAA}$.
\end{description}
\end{definition}

Conditions \textbf{Lock-1} and \textbf{Lock-2} together state that a
correct lock admits an alternating sequence of acquire and release
operations, where each acquire and release pair is from the same
thread. Locks do not have to be released at the end of a trace.

Condition \textbf{Fork-1} states that a thread can be created at most once.
Condition \textbf{Fork-2} states that each thread except the main thread
is preceded by a fork event.
Both conditions imply that for each event ${(\beta, t,\forkE{s})}$ we have $t \ne s$.

Condition \textbf{Join} states that all events from a joined thread appear
before the join event.
There can be several join events $\joinE{t}$ for the same thread
$t$. A join operation $\joinE{t}$ does not necessarily
have to appear in the thread that forked thread $t$.

\begin{figure}
  \centering
  \bda{|l|l|l|l|} \hline & \thread{1} & \thread{2} & \thread{3}\\
  \hline
  \eventE{1}  & \forkE{\thread{2}}&&\\
  \eventE{2}  & \forkE{\thread{3}}&&\\
  \eventE{3}  & \lockE{x}&&\\
  \eventE{4}  & &\lockE{y}&\\
  \eventE{5}  & &\unlockE{y}&\\
  \eventE{6}  & &&\lockE{z}\\
  \eventE{7}  & &&\joinE{\thread{2}}\\
  \eventE{8}  & &&\unlockE{z}\\
  \eventE{9}  & \joinE{\thread{2}}&&\\
  \eventE{10}  & \unlockE{x}&&\\
  \hline \eda{}
 \caption{Example of well-formed trace}
 \label{fig:ex-well-formed}
 \end{figure}

\begin{example}
  The trace in figure~\ref{fig:ex-well-formed} is well-formed as all conditions of
  definition~\ref{def:well-formed-trace} are met. The event
  $\joinE{\thread{2}}$ appears twice, which is unusual, but accepted
  by our definition.

\end{example}

\noindent
{\bf Trace Reordering and Must-Happen-Before.}
A trace represents one possible interleaving of concurrent events.
In theory, there can be as many interleavings as there are permutations of the original trace.
However, not all permutations are feasible in the sense that they
could be reproduced by executing the program with a different schedule.
In addition to well-formedness, a reordering must guarantee that (a) the program order
and (b) last writes are maintained.
A reordering maintains program order if the order of events within any thread remains the same.
A reordering maintains last writes if any read observes the same write event.
The latter ensures that every read obtains the same value, so that the
control flow of the program remains the same.
We now  formalize the criteria for correct reorderings.

From now on, we assume that $T$ is a well-formed trace.
The \emph{projection} of $T$ onto thread $t$ is the trace $T' =
\proj{t}{T}$ consisting of all events $e\in T$ with $\thd{e} = t$ in the same order as in $T$. That is, (1) for
each $\evtAA = (t,op) \in T$ we have that $\evtAA \in T'$, and
(2) for each $\evtAA, \evtBB \in T'$,
$\posP{T'}{\evtAA} < \posP{T'}{\evtBB}$
implies $\posP{T}{\evtAA} < \posP{T}{\evtBB}$.

We define $\rwTa$ as the set of all read/write
events in $T$ on some shared variable $a$.
We define $\rwTr$ as the union of $\rwTa$ for all shared variables $a$.

Take events $\evtAA = (t, \readE{a}), \evtBB = (s, \writeE{a}) \in \rwTa$.
We say that $\evtBB$ is the \emph{last write} for $\evtAA$ w.r.t.~$T$ if
(1) $\ltTr{\evtBB}{\evtAA}$ , and
(2) there is no event $\evtCC = (u, \writeE{a}) \in T$ such that
    $\ltTrr{\evtBB}{\evtCC}{\evtAA}$.
That is, the last write appears before the read with no other write in between.

\begin{definition}[Correct Trace Reordering]
  \label{def:correct-trace-reordering}
Trace $T'$ is a \emph{correctly reordered prefix} of $T$ if the following
conditions are satisfied:
\begin{description}
\item[WF:] Trace $T'$ is well-formed and $\events{T'} \subseteq \events{T}$.

\item[PO:] For each thread id~$t$, $\proj{t}{T'}$ is a prefix of $\proj{t}{T}$.

\item [LW:] For each read event $\evtAA = (\alpha,t,\readE{a}) \in T'$
            where $\evtBB = (\beta, s, \writeE{a})$ is the last write for $\evtAA$ w.r.t.~$T$,
            it must be $\evtBB \in T'$ and
            $\evtBB$ is also the last write for $\evtAA$
            w.r.t.~$T'$.\footnote{Unique event identifiers are crucial for this condition.}
\end{description}
\end{definition}

By considering all (correct) trace reorderings that are derivable from a given trace,
we can determine if an event \emph{must happen  before} another event.

\begin{definition}[Must-Happen-Before Relation]
  \label{def:must-happen-before}
  The \emph{must-happen-before relation}
  $\ltMustHB{T}{}{}$ is a  binary relation on
  $\events{T}$ such that for all distinct events $e, f \in
  T$,   $\ltMustHB{T}{e}{f}$ if
  for all correctly reordered traces $T'$  of $T$ such that $e, f \in
  T'$ we have
  $\ltTrace{T'}{e}{f}$.

  We write $\concMustHB{T}{e}{f}$ if neither $\ltMustHB{T}{e}{f}$ nor
  $\ltMustHB{T}{f}{e}$ holds.

  We write $\ltMHB{\evtAA}{\evtBB}$ if trace $T$ is determined
  by the context.
\end{definition}

\begin{lemma}
  The relation $\ltMustHB{T}{}{}$ is a strict partial ordering on $\events{T}$.
\end{lemma}
\begin{proof}
  Irreflexivity is immediate.
  For transitivity assume that $\ltMustHB{T}{e}{f}$ and
  $\ltMustHB{T}{f}{g}$.
  That is, for all correctly reordered traces $T'$, $\ltTrace{T'}{e}{f}$ and
  $\ltTrace{T'}{f}{g}$.  Transitivity of $\ltTrace{T'}{}{}$
  yields $\ltTrace{T'}{\evtAA}{\evtCC}$, which implies $\ltMustHB{T}{\evtAA}{\evtCC}$.
\end{proof}


\noindent
{\bf Predictable data races.}
To talk about data races, we first define when two events are in
conflict: they have to take place in different threads, they must refer
to read/write operations, and one of them must be a write.

    We generally assume that $T$ is a well-formed
trace. We omit indices if the trace $T$ is clear from
the context.

\begin{definition}
  Let $\evtAA = (s, op), \evtBB = (t, op') \in T$. Events $\evtAA$ and
  $\evtBB$ are \emph{conflicting} (notation: $\conflict\evtAA\evtBB$) if $s \ne t$ and there exists some variable
  $a$ such that
  $op = \preRead{a}$ and $op' = \preWrite{a}$ or
  $op = \preWrite{a}$ and $op' = \preRead{a}$ or
  $op = \preWrite{a}$ and $op' = \preWrite{a}$.
\end{definition}
Conflicting events are harmless as long as they are properly
ordered. If they are not sufficiently constrained, we predict a data race.
\begin{definition}[Predictable Data Race]
  Let $e, f \in T$ be two conflicting events.

  We say that $(e, f)$ are in a \emph{predictable data race}
  (notation $\DataRace{T}{e}{f}$) if
  there exists a correctly reordered prefix $T'$ of $T$ such that
  $e$ and $f$ appear right next to each other in $T'$.

  Let $\TrueDR[T] = \{ (e, f) \mid \DataRace{T}{e}{f} \}$ be the
  set of predictable data races.
\end{definition}


The set $\TrueDR[T]$ serves as the ground truth for dynamic
data race analysis. It is computable, but no efficient algorithm is known.


\section{Cross-Thread Critical Sections}
\label{sec:cross-thread}

Informally, an event is part of a critical section if it must happen
between an acquire of a lock
and its matching release.
In this case we say that the event \emph{holds the lock} with the
understanding that an event can hold more than one lock if it is
bracketed by
more than one matching pair of acquire and release operations. This observation gives
rise to the notion of the lock set of an event.

We start this section by defining critical sections according to their
standard use in the literature, point out that the fixation on a
single thread is an intrinsic shortcoming of this definition,
and propose alternative definitions that address this shortcoming.

We assume that $T$ is a well-formed trace with a matching release
event for each acquire.\footnote{This assumption imposes no
  restriction as we can always add missing release events at the
  end of a trace.}
Most definitions are indexed with a trace $T$; we omit the index if
it is clear from the context.

\subsection{Standard Critical Sections}
\label{sec:stand-crit-sect}

While the notion of a critical section goes back to Dijkstra
\cite{DBLP:journals/cacm/Dijkstra65}, it is hard to find a
mathematical definition.
The following definition captures the usual understanding in the literature.
\begin{definition}[Standard Critical Section]
\label{def:std-cs}
  Suppose there are events
  $\evtAcc = (t, \lockE{x}), \evtRel = (t, \unlockE{x}) \in T$ for
  some lock $x$.
  We say that $e \in T$ is in the
  \emph{standard critical section for lock $x$ guarded by acquire $\evtAcc$
    and release $\evtRel$},
  written $e \in \StdCSect{T}{\AcqRelPair{\evtAcc}{\evtRel}}{x}$ if
  the following three conditions hold:

  \begin{description}
  \item[CS-BRACKET:]  $\ltMustHBB{T}{\evtAcc}{e}{\evtRel}$.

  \item[CS-MATCH:] There is no release event $\evtRelA= (t, \unlockE{x})
    \in T$  and
     correct reordering $T'$ of $T$
     such that $\ltTraceT{T'}{\evtAcc}{\evtRelA}{e}$.

  \item[CS-SAME-THREAD:] $e = (\alpha, t,op)$ for some operation $op$ in the
    same thread as $\evtAcc$ and $\evtRel$.

  \end{description}
\end{definition}

We call $\evtAcc, \evtRel$  a \emph{matching acquire-release pair}.
We omit guard events if they are clear from the context.
Condition \textbf{CS-BRACKET} states that the acquire-release pair $\evtAcc$ and $\evtRel$ must happen before and after~$e$.
Condition \textbf{CS-MATCH} guarantees that $\evtRel$ is the matching release for the acquire event $\evtAcc$.
Condition \textbf{CS-SAME-THREAD} states that $e$ is in the same thread as the acquire-release pair.
%



Lock $x$ is held by some event $e$ if the event is part of a critical section for lock~$x$.

\begin{definition}[Standard Locks Held]
  The \emph{standard lock set for an event $e\in T$} is defined by
  $\StdLHeld{T}{e} =
  \{ x \mid \exists \evtAcc, \evtRel \in T.
        e \in \StdCSect{T}{\AcqRelPair{\evtAcc}{\evtRel}}{x} \}$.
\end{definition}

\subsection{Cross-Thread Critical Sections}
\label{sec:cross-thre-crit}

Definition~\ref{sec:stand-crit-sect}
does not capture all cases where an event must happen between an acquire-release pair because the definition insists that the event
must be in the same thread as the acquire-release.
As observed in the overview section in figure~\ref{fig:exCS_2}
there are  acquire-release pair that guards events across thread boundaries.


Hence,
we relax the notion of a critical section to include events from
other threads simply by dropping condition \textbf{CS-SAME-THREAD}.

\begin{definition}[Cross-Thread Critical Section]
  \label{def:cross-thread-cs}
  Suppose there are events
  $\evtAcc = (t, \lockE{x}), \evtRel = (t, \unlockE{x}) \in T$ for
  some lock $x$.
  We say that $\evtSubject \in T$ is in the
  \emph{cross-thread critical section for lock $x$ guarded by acquire
    $\evtAcc$ and release $\evtRel$},
  written $\evtSubject \in \IntraCSect{T}{\AcqRelPair{\evtAcc}{\evtRel}}{x}$ if
  the following two conditions hold:

  \begin{description}
  \item[CS-CROSS-1:]  $\ltMustHBB{T}{\evtAcc}{\evtSubject}{\evtRel}$.

  \item[CS-CROSS-2:] There is no release event $\evtRelA = (t, \unlockE{x}) \in T$ and
     correct reordering $T'$ of $T$
     such that $\ltTraceT{T'}{\evtAcc}{\evtRelA}{\evtSubject}$.

  \end{description}
\end{definition}

Again, a lock $x$ is held by an event $e$ if $e$ is in a cross-thread
critical section for $x$.

\begin{definition}[Cross-Thread Locks Held]
  The \emph{cross-thread lock set for an event $e \in T$} is defined by $\IntraLHeld{T}{e} =
  \{ x \mid \exists \evtAcc, \evtRel \in T.
  e \in \IntraCSect{T}{\AcqRelPair{\evtAcc}{\evtRel}}{x} \}$.
\end{definition}

We revisit the examples from the overview section~\ref{sec:overview}
and compare the critical sections and the lock sets of the standard
construction with the ones of the cross-thread construction.

\begin{example}
  \label{ex:exCS_2}
    Recall the trace in figure~\ref{fig:exCS_2}.
    We find that $\StdCSec{\AcqRelPair{e_7}{e_9}}{x} = \IntraCSec{\AcqRelPair{e_7}{e_9}}{x} = \{ e_8 \}$
    but $\StdCSec{\AcqRelPair{e_2}{e_{6}}}{x} = \{ e_3, e_{5} \}$
    and $\IntraCSec{\AcqRelPair{e_2}{e_{6}}}{x} = \{ e_3, e_4, e_5 \}$.
    Hence, $\StdLH{e_8} = \IntraLH{e_8} = \{x\}$ but
     $\StdLH{e_{4}} = \{ \}$ and $\IntraLH{e_{4}} = \{ x \}$.
  \end{example}


\subsection{Properties}
\label{sec:properties}

It is easy to see that each of the relaxations of the definition of a
critical section potentially increases the lock set of each event.

\begin{lemma}
  \label{le:lh}
  \begin{enumerate}
  \item For each $e \in T$, we have
    $\StdLHeld{T}{e} \subseteq \IntraLHeld{T}{e}$.
  \item There exist traces $T$ and $e \in T$, such that
    $\StdLHeld{T}{e} \subsetneq \IntraLHeld{T}{e}$.
  \end{enumerate}
\end{lemma}

\begin{proof}
  \begin{enumerate}
  \item The inclusions are immediate from the definitions of the
    different critical sections.
  \item
    Example~\ref{ex:exCS_2} shows an event where the
    two lock sets are different.
    \qedhere
  \end{enumerate}
\end{proof}


\section{Data Race Prediction Based on Lock Sets}
\label{sec:lockset-race-pred}

Lock sets can be exploited for data race prediction.
To cater for the different lock set constructions in
section~\ref{sec:cross-thread},
we define an abstract data race predictor for trace $T$ parameterized by a
\emph{lock set function} $L$ that maps $\events{T}$ to sets of lock
variables.
We predict a
datarace if there is a pair of conflicting events in a trace with
disjoint lock sets.
\begin{align*}
  \PredictDR[T]{L} & := \{ (e, f) \mid e, f \in T, \conflict ef, L(e)
                     \cap L(f) = \emptyset \}
\end{align*}
Given any lock set function $L$, we can define false positives and
false negatives of the corresponding predictor by comparing it with
the ground truth given by $\TrueDR[T]$.\footnote{The ground truth is
  usually unknown, but we can nevertheless use it in a mathematical definition.}
\begin{definition}[Data Race Prediction False Positives/False Negatives]
  Let $L$ be a lock set function.
  We define $\FalseP{T}{L}$, the set of false positives, and
  $\FalseN{T}{L}$, the set of false negatives for data race prediction
  based on $L$.
\begin{align*}
  \FalseP{T}{L} & := \PredictDR[T]{L} \setminus \TrueDR[T] && \text{predicted, but no
                                              data race} \\
  \FalseN{T}{L} & := \TrueDR[T] \setminus \PredictDR[T]{L} &&
                                                              \text{data race,
                                                    but not predicted}
\end{align*}
\end{definition}
Next, we compare data race predictors for lock set
functions $L_1 \subseteq L_2$ of different precision.\footnote{For set-valued mappings $L_1, L_2$, we lift the subset relation
$\subseteq$ pointwise, that is, we write $L_1 \subseteq L_2$ if, for all $e$,
$L_1(e) \subseteq L_2(e)$.} Smaller lock sets lead to larger numbers
of predicted dataraces and thus to a potentially larger number of
false positives. On the other hand, smaller lock sets result in fewer
false negatives.

\begin{lemma}
  \label{le:fp-fn}
  Let $L_1, L_2$ be lock set functions with  $L_1 \subseteq L_2$.
  \begin{enumerate}
  \item $\PredictDR[T]{L_2} \subseteq \PredictDR[T]{L_1}$,
  \item $\FalseP{T}{L_2} \subseteq \FalseP{T}{L_1}$,
  \item $\FalseN{T}{L_1} \subseteq \FalseN{T}{L_2}$.
  \end{enumerate}
\end{lemma}

For the proof, we recall a basic result from
set theory.
\begin{lemma}
  \label{le:set}
  Let $A, B, C, D$ be sets with $A \subseteq B$ and $C \subseteq D$.
  Then, $A \cap C \subseteq B \cap D$.
\end{lemma}

\begin{proof}[Proof of lemma~\ref{le:fp-fn}]~
  \begin{enumerate}
  \item If $(e, f) \in \PredictDR[T]{L_2}$, then $\conflict ef$ and
    $L_2(e) \cap L_2(f) = \emptyset$.
    By lemma~\ref{le:set}, $L_1(e) \cap L_1(f) = \emptyset$, so that
    $(e, f) \in \PredictDR[T]{L_1}$.
  \item $\FalseP{T}{L_2} = \PredictDR[T]{L_2} \setminus \TrueDR[T]
    \stackrel{\mathrm{item}~1}{\subseteq} \PredictDR[T]{L_1} \setminus
    \TrueDR[T]  = \FalseP{T}{L_1}$.
  \item $\FalseN{T}{L_1} = \TrueDR[T] \setminus \PredictDR[T]{L_1}
    \stackrel{\mathrm{item}~1}\subseteq \TrueDR[T] \setminus
    \PredictDR[T]{L_2} =  \FalseN{T}{L_2}$.
    \qedhere
  \end{enumerate}
\end{proof}

Each of the lock set constructions in
section~\ref{sec:cross-thread} gives rise to a data race predictor by instantiating the lock set function
$L$ accordingly. From lemmas~\ref{le:lh} and~\ref{le:fp-fn}, we obtain several corollaries.
We show that the predictor based on the standard construction
$\StdLH{\cdot}$ has no false negatives (lemma~\ref{le:std-has-no-false-negatives}).
The cross-thread predictor based on $\IntraLH{\cdot}$ may give rise to false
negatives  (Corollary~\ref{co:fn}), but they can be amended as shown in Section~\ref{sec:elim-more-false}.
On the other hand, the predictor based on $\IntraLH{\cdot}$ reports
fewer races and has fewer false positives as $\StdLH{\cdot}$
(Corollaries~\ref{co:predict} and~\ref{co:fp}).
\begin{corollary}
  \label{co:predict}
  $\PredictDR{\IntraLH{\cdot}} \subseteq \PredictDR{\StdLH{\cdot}}$.
\end{corollary}

\begin{corollary}
  \label{co:fp}
  $\FP{\IntraLH{\cdot}} \subseteq \FP{\StdLH{\cdot}}$.
\end{corollary}

\begin{corollary}
  \label{co:fn}
  $\FN{\StdLH{\cdot}} \subseteq \FN{\IntraLH{\cdot}}$.
\end{corollary}
%

Corollary~\ref{co:fn} also follows from the fact that the standard lock set
construction has no false negatives.

\begin{lemma}
  \label{le:std-has-no-false-negatives}
  $\FN{\StdLH{\cdot}} = \emptyset$.
\end{lemma}
\begin{proof}

      Suppose that $(e, f) \in \FN{\StdLH{\cdot}}$.

      That is, $\conflict ef$, $\DataRace{T}{e}{f}$, but there exists
      some lock $x$ such that $x \in \StdLH{e}
      \cap \StdLH{f} \ne \emptyset$.

      By definition of $\StdLH{}$ there must be an acquire event
      $\evtAcc \in T$
      for lock $x$
      such that $ \ltMustHB{T}{\evtAcc}{e}$
      as well as an acquire event $\evtAcc' \in T$ for
      lock $x$ such that  $ \ltMustHB{T}{\evtAcc'}{f}$.
      From \textbf{CS-SAME-THREAD} we know that events $\evtAcc$ and $e$ are in the
      same thread and that events $\evtAcc'$ and $f$ are in the same thread.

      As $\DataRace{T}{e}{f}$, there is a correctly
      reordered prefix $T'$ of $T$ that puts $e$ and $f$ next to each other.
      As $\evtAcc$ and $\evtAcc'$ must appear before $e$ and $f$ in $T'$
      without an intervening release (\textbf{CS-MATCH}), we find a
      contradiction as $T'$ would violate \textbf{Lock-1}.
\end{proof}

Datarace prediction based on the cross-thread lock set construction gives rise to false negatives,
but we eliminate this problem in Section~\ref{sec:elim-more-false} by
improving the datarace predictor.
\begin{lemma}
  \label{le:ex9}
  There exists a trace $T$ and event $e\in T$ such that $\FN{\IntraLH{e}} \ne \emptyset$.
\end{lemma}
\begin{proof}
  \begin{minipage}[t]{0.5\linewidth}
    Consider the trace on the right. Here $\IntraLH{\eventE{3}} =
    \IntraLH{\eventE{4}} = \{ x \}$ so that
    $\PredictDR{\IntraLH{\cdot}} = \emptyset$, although there is
    clearly a data race $\DataRace{}{\eventE{3}}{ \eventE{4}} $.
  \end{minipage}
  \quad
  \begin{minipage}[t]{0.4\linewidth}
    \vspace{-\baselineskip}
\bda{|l|l|l|}
\hline  & \thread{1} & \thread{2}\\ \hline
\eventE{1}  & \lockE{x}&\\
\eventE{2}  & \forkE{\thread{2}}&\\
\eventE{3}  & \writeE{a}&\\
\eventE{4}  & &\writeE{a}\\
\eventE{5}  & \joinE{\thread{2}}&\\
\eventE{6}  & \unlockE{x}&\\

 \hline \eda{}
\end{minipage}
\end{proof}









\subsection{Eliminating false negatives by thread indexing}
\label{sec:elim-more-false}

Using insights from the proof of
Lemma~\ref{le:std-has-no-false-negatives}, we refine the cross-thread lock set construction
to elide all false negatives. The key step in the proof
is the observation that the acquire (and release) events guarding the
conflicting events $e$ and $f$ are in the same thread as $e$ and $f$
by \textbf{CS-SAME-THREAD}, respectively. The cross-thread lock set
construction omits the axiom \textbf{CS-SAME-THREAD}, so that the guarding
events may end up in the same thread and coincide as in
the proof of lemma~\ref{le:ex9}.

We rule out such cases by demanding that the guarding events are in distinct threads.
For this purpose, we modify the cross-thread lock set
construction to retain the information which thread acquired each
lock. Consequently, the resulting lock set function maps $\events{T}$ to a set of
thread-indexed lock variables $(x,t)$.

\begin{definition}[Cross-Thread Locks Held by Thread]
  \label{def:thread-indexed-lock-set}
  The \emph{cross-thread thread-indexed lock set for an event $e \in T$} is defined by $\CTTLHeld{T}{e} =
  \{ (x, t) \mid \exists \evtAcc, \evtRel \in T.
  e \in \IntraCSect{T}{\AcqRelPair{\evtAcc}{\evtRel}}{x}, t = \thd{\evtAcc} \}$.
\end{definition}

We have to adapt the predictor to process this lock set construction.
For thread-indexed lock sets $M'$ and $N'$ the intersection
only includes lock variables with different indices, that is, locks
that have been acquired in different threads.
\begin{align*}
  M' \indexedcap N' & := \{ x \mid (x, s) \in M', (x, t) \in N', s \ne t \}
\end{align*}

We define the predictor for indexed lock sets in terms of this operation:
\begin{align*}
  \PredictDRP[T]{L'} & := \{ (e, f) \mid e, f \in T, \conflict ef, L'(e)
                     \indexedcap L'(f) = \emptyset \}
\end{align*}

\begin{lemma}\label{lemma:indexedcap}
  $ M' \indexedcap N' \subseteq \pi_1(M') \cap \pi _1( N')$ (where $\pi_1$
  denotes projection on the first component).
\end{lemma}
\begin{proof}
  Suppose $x \in  M' \indexedcap N'$.
  Then there are $s, t$ such that $(x, s) \in M'$ and $(x,t) \in N'$ and
  $s \ne t$.
  Hence, $x \in \pi_1(M') \cap \pi_1(N')$.
\end{proof}
The reverse inclusion does not hold in general: if $M' = N' = \{ (x,
1) \}$, then $M' \indexedcap N' = \emptyset$, but $\pi_1(M')\cap
\pi_1(N') =  \{x\}$.

The predictions of the new thread-indexed cross-thread construction fit in between the
standard construction and the cross-thread construction.
The corresponding formal statements follow from the definitions.
\begin{lemma}~\\[-\baselineskip]
  \label{le:thread-indexed-relations}
  \begin{enumerate}
  \item For all events $e \in T$, $\pi_1(\CTTLHeld{T}{e}) =
    \IntraLHeld{T}{e}$.
  \item For all events $e \in T$, $\StdLHeld Te = \{ x \mid (x, t) \in
    \CTTLHeld Te, t = \thd e \}$.
  \item $\PredictDR[T]{\IntraLH{\cdot}} \subseteq \PredictDRP[T]{\CTTLH{\cdot}}$
  \qquad(old predictor on the left
  and new one on the right).
\item $\PredictDRP[T]{\CTTLH{\cdot}} \subseteq\PredictDR[T]{\StdLH{\cdot}} $
  \qquad (new predictor on the left and old one on the right).
  \end{enumerate}
\end{lemma}
\begin{proof}
  \begin{enumerate}
  \item Immediate from the definition.
  \item Follows from axiom \textbf{CS-SAME-THREAD}: the locks in the
    standard lock set are acquired by the event's thread.
  \item Suppose that $(e, f) \in \PredictDR[T]{\IntraLH{\cdot}}$.
    That is, $\conflict ef, \IntraLH e \cap \IntraLH f = \emptyset$.
    By item~(1), $\pi_1 (\CTTLH e) \cap \pi_1 (\CTTLH f) = \emptyset$.
    By lemma~\ref{lemma:indexedcap}, $\CTTLH e \indexedcap \CTTLH f =
    \emptyset$.
    Conclude $(e,f) \in \PredictDRP[T]{\CTTLH{\cdot}}$.
  \item Suppose that $(e,f) \in \PredictDRP[T]{\CTTLH{\cdot}}$ with $s
    = \thd e$ and $t = \thd f$.
    That is, $\conflict ef, \CTTLH e \indexedcap \CTTLH f =
    \emptyset$.
    Suppose now for a contradiction that there exists some $x \in \StdLH e \cap \StdLH f$.
    From item~(2), we see that $(x, s) \in \CTTLH e$ and $(x,t) \in
    \CTTLH f$ and from $\conflict ef$ we know that $s \ne t$.
    We obtain $x \in \CTTLH e \indexedcap \CTTLH f = \emptyset$, which is
    a contradiction.
    Hence, $\StdLH e \cap \StdLH f = \emptyset$ and $(e,f) \in \PredictDR[T]{\StdLH{\cdot}}$.
  \end{enumerate}
\end{proof}

\begin{figure}[tp]
  \begin{center}
    \begin{tikzpicture}
      \node at (0,0) [rectangle,draw=black,fill=blue!20,thick,minimum width=50ex,minimum height=35ex] {
        \tikz{
          \node at (0,0.5) [draw=black,fill=orange!60,minimum width=35ex,minimum height=25ex,text width=34ex,text depth=24ex] {$\PredictDR{\StdLH{\cdot}}$};
          \node at (0,0.25) [draw=black,fill=orange!40,minimum width=30ex,minimum height=20ex,text width=29ex,text depth=19ex] {$\PredictDRP{\CTTLH{\cdot}}$};
          \node at (0,0) [draw=black,fill=orange!20,minimum width=25ex,minimum height=15ex,text width=24ex,text depth=14ex] {$\PredictDR{\IntraLH{\cdot}}$};
          \node at (0.5,-0.3) [draw=black,fill=yellow!20,minimum width=23ex,minimum height=12ex] {$\TrueDR$};
        }
        };
      \end{tikzpicture}
  \end{center}
    \caption{Predictions for different data race predictors relative to ground truth}
  \label{fig:fp-fn-containment}
\end{figure}
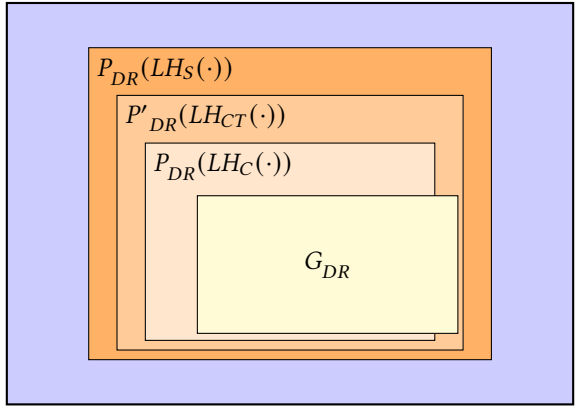


Given the inclusions between the different datarace predictors, the
inclusions for false positives and false negatives from
Corollaries~\ref{co:fn} and~\ref{co:fp} extend accordingly.
Moreover, the thread-indexed cross-thread construction has no false negatives.
\begin{corollary}~\\[-\baselineskip]
  \label{co:ct-inclusion}
  \begin{itemize}
  \item
    $\FN{\StdLH{\cdot}} \subseteq \FN{\CTTLH{\cdot}} \subseteq \FN{\IntraLH{\cdot}}$.
  \item
    $\FP{\IntraLH{\cdot}} \subseteq \FP{\CTTLH{\cdot}} \subseteq \FP{\StdLH{\cdot}}$.
  \end{itemize}
\end{corollary}

\begin{corollary}
  \label{co:ct-no-fn}
  $\FN{\CTTLH{\cdot}} = \emptyset$.
\end{corollary}
\begin{proof}
  The proof of absence of false negatives in
  lemma~\ref{le:std-has-no-false-negatives} extends to
  $\PredictDRP[T]{\CTTLH{\cdot}}$.
\end{proof}


Figure~\ref{fig:fp-fn-containment} gives an overview of the relative
power of all variations of datarace predictors considered in this
section. It visualizes the Corollaries~\ref{co:predict}, \ref{co:fp}, and
\ref{co:fn} (extended with Corollaries~\ref{co:ct-inclusion}
and~\ref{co:ct-no-fn}), as well as
lemma~\ref{le:std-has-no-false-negatives}.

\subsection{On False Positives}
\label{sec:false-positives}

There are examples where the inclusion in Corollary~\ref{co:fp} is proper.

\begin{example}
  The standard predictor yields a false positive for the trace in
  figure~\ref{fig:exCS_2}, which is
  captured by the cross-thread predictor: $\StdLH{\eventE4} =
  \emptyset$ and  $\StdLH{\eventE8} = \{x\}$, but
  $\IntraLH{\eventE4} =\IntraLH{\eventE8} = \{x\}$.
  This false positive is also captured by the thread-indexed cross-thread
  predictor because the locks for $\eventE4$ and $\eventE8$ are
  acquired in different threads.
\end{example}


\section{PWR and Cross-Thread Critical Sections}
\label{sec:pwr}

In this section, we describe an algorithm for data race detection that
takes advantage of cross-thread critical sections to calculate lock
sets. It relies on the PWR relation to efficiently compute an
underapproximation of the must happen-before relation, which in turn
results in an underapproximation (smaller sets) of the cross-thread lock sets.


\subsection{Program, Last-Write and Release Order Relation}

We first recall the original definition of PWR, which relies on the
standard lock set construction.

\begin{definition}[Program, Last Write, Release Order (PWR)~\cite{10.1145/3426182.3426185}]
  \label{def:www-relation}
  For a well-formed trace $T$,
  the \emph{program-order, last write-order, release-order relation}
  $\ltPWR{T}{}{}$
  is the smallest strict partial order that
  satisfies the following rules:

  \begin{description}
  \item[\nf{(PWR-1)}:]
        $\ltPWR{T}{\evtAA}{\evtBB}$ if
       $\ltTrace{T}{\evtAA}{\evtBB}$ and $\thd{\evtAA} = \thd{\evtBB}$.

  \item[\nf{(PWR-2)}:]
      $\ltPWR{T}{\evtBB}{\evtAA}$ if
        $\evtAA = (t,\readE{a})$ and
        $\evtBB$ is the last write for $\evtAA$ w.r.t.~$T$.
  \item[\nf{(PWR-3)}:]
    $\ltPWR{T}{r_1}{f}$ if
    for any two matching acquire-release pairs $(a_1,r_1)$ and $(a_2,r_2)$
    and for some lock $x$ where $\ltTrace{T}{a_1}{a_2}$
    there is some event $e$ such that
    $e \in \StdCSect{T}{\AcqRelPair{a_1}{r_1}}{x}$ and
    $f \in \StdCSect{T}{\AcqRelPair{a_2}{r_2}}{x}$.
  \item[\nf{(PWR-4)}]
    $\ltPWR{T}{\evtAA}{\evtBB}$ if $\evtAA = (t, \forkE{s})$ and $s = \thd{\evtBB}$,

  \item[\nf{(PWR-5)}]
    $\ltPWR{T}{\evtAA}{\evtBB}$ if $\evtBB = (t, \joinE{s})$ and $s = \thd{\evtAA}$.
  \end{description}

Conflicting events $\conflict\evtAA\evtBB$
are in a \emph{PWR-race} if $\concPWR{T}{\evtAA}{\evtBB}$
and $\StdLH{\evtAA} \cap \StdLH{\evtBB} = \emptyset$.

\end{definition}
Rule (PWR-1) imposes program order and rule (PWR-2) guarantees that each read sees the same write.
Rule (PWR-3) imposes an order among critical sections.

PWR is a hybrid method as a race warning is only issued
for conflicting events that are concurrent and have disjoint lock sets.

PWR underapproximates the must happen-before relation in the following sense.

\begin{lemma}\label{lemma:pwr-lt-cross-thread}
  (cf.\ \cite[Proposition 3.5]{10.1145/3426182.3426185})
  $\ltPWR{T}{{}}{{}} \subseteq \ltMustHB{T}{{}}{{}}$.
\end{lemma}
The example in figure~\ref{fig:exDC_2b} demonstrates that
the reverse inclusion does not hold:
We have $\ltMustHB{T}{e_5}{e_{10}}$, but PWR does not order these events.

From Lemma~\ref{le:std-has-no-false-negatives} and
Lemma~\ref{lemma:pwr-lt-cross-thread} we obtain
that PWR is a complete data race prediction method (no false
negatives).
However, we may encounter false positives.
Consider the example in figure~\ref{fig:exCS_2}.
Conflicting events $e_4$ and $e_8$ are not ordered under PWR and their standard lock set is disjoint.

Next, we show how to make use of the PWR relation to compute an
approximation of cross-thread lock sets. The resulting construction
retains completeness and eliminates some false positives (like the one
discussed in the previous paragraph).

\subsection{Approximating Cross-Thread Lock Sets via PWR}
\label{sec:effic-comp-cross}

\begin{figure}

\bda{|l|l|l|l|l|}
\hline  & \thread{1} & \thread{2} & \thread{3} & \thread{4}\\ \hline
\eventE{1}  & \forkE{\thread{2}}&&&\\
\eventE{2}  & \forkE{\thread{3}}&&&\\
\eventE{3}  & \lockE{x}&&&\\
\eventE{4}  & \writeE{a}&&&\\
\eventE{5}  & &\lockE{y}&&\\
\eventE{6}  & &\writeE{b}&&\\
\eventE{7}  & \joinE{\thread{2}}&&&\\
\eventE{8}  & \unlockE{x}&&&\\
\eventE{9}  & &&\lockE{x}&\\
\eventE{10}  & &&\forkE{\thread{4}}&\\
\eventE{11}  & &&&\writeE{b}\\
\eventE{12}  & &&\readE{a}&\\
\eventE{13}  & &&\unlockE{x}&\\
\eventE{14}  & &\joinE{\thread{4}}&&\\
\eventE{15}  & &\unlockE{y}&&\\

 \hline \eda{}

   \caption{$\PWRLH{\cdot}$ underapproximates $\IntraLH{\cdot}$: $(y,
     t_2) \in \CTTLH{e_{11}}$ but $(y,t_2) \not\in \PWRLH{e_{11}}$}
\label{fig:exPWR_1}

\end{figure}

We define a new lock set function $\PWRLH{\cdot}$ as a variant of $\CTTLH{\cdot}$ which
replaces the relation $\ltMHB{\cdot}{\cdot}$ in definition~\ref{def:cross-thread-cs} with the
relation $\ltPWR{}{\cdot}{\cdot}$ and then calculates a thread-indexed
lock set analogous to
definition~\ref{def:thread-indexed-lock-set}. This new construction
lies properly between the standard and cross-thread lock set construction.

\begin{lemma}\label{lemma:lockset-s-pwr-c}
  \begin{enumerate}
  \item For all $e\in T$, $\PWRLH{e} \subseteq \CTTLH{e}$.
  \item For all $e\in T$, $\StdLH{e} = \{ x \mid (x, t) \in \PWRLH{e},
    t = \threadID e\} $.
  \item There exists a trace $T$ and $e\in T$, such that $\PWRLH{e}
    \subsetneq \CTTLH{e}$.
  \item There exists a trace $T$ and $e\in T$, such that $\StdLH{e}
    \subsetneq \pi_1 (\PWRLH{e})$.
  \end{enumerate}
\end{lemma}
\begin{proof}
  \begin{enumerate}
  \item The inclusion $\PWRLH{\cdot} \subseteq \CTTLH{\cdot}$ follows
    from lemma~\ref{lemma:pwr-lt-cross-thread}.
  \item Suppose that $x \in \StdLH{e}$, that is $\ltMHBB{a(x)}
    e{r(x)}$ and \textbf{CS-SAME-THREAD} and \textbf{CS-MATCH}. Hence,
    ${a(x)} <_{po} e <_{po} {r(x)}$ in the program order and
    \textbf{CS-MATCH}. By (PWR-1), ${a(x)} <_{\pwrsymbol} e <_{\pwrsymbol} {r(x)}$
    and \textbf{CS-MATCH}, so that $(x, \threadID e) \in \PWRLH (e)$.
  \item
    Figure~\ref{fig:exPWR_1} contains a trace where
    $\ltMustHB{T}{e_6}{e_{11}}$ via similar reasoning as in
    figure~\ref{fig:exDC_2b}.
    Due to the join event $e_{14}$, event $e_{11}$ is protected by
    lock $y$. That is, $(y, t_2) \in \IntraLH{e_{11}}$.
    However, $(y, t_2) \not\in \PWRLH{e_{11}}$.
  \item
    In figure~\ref{fig:exCS_2}, $\StdLH{e_4} = \emptyset$, but
    $\CTTLH{e_4} = \{ (x, t_1) \}$.
    \qedhere
  \end{enumerate}
\end{proof}

Together with Lemma~\ref{le:thread-indexed-relations}, we
see that the difference between $\CTTLH{e}$ and $\PWRLH{e}$ lies only
in threads other than $t =\threadID e$, as the projection on thread
$t$ yields $\StdLH{e}$ in both cases. Applying
corollary~\ref{co:ct-no-fn} to lemma~\ref{lemma:lockset-s-pwr-c} we
obtain completeness for data race prediction with PWR lock sets:
\begin{corollary}
  \label{le:pwr-has-no-false-negatives}
  $\FN{\PWRLH{\cdot}} = \emptyset$.
\end{corollary}

The implementation of the enhanced PWR algorithm only caches the most
recent reads and writes for efficiency reasons.
For example, conflicting events $e_2$ and $e_7$ in figure~\ref{fig:exPWR_2}
are in a race as shown by the reordering on the right.
However, the subsequent event $e_4$ will evict $e_2$ and we will miss the race.
This strategy (re)introduces some incompleteness, but it is also
applied by other implementations~\cite{flanagan2010fasttrack}.

\begin{figure}

  \bda{lcl}

\ba{|l|l|l|}
\hline  & \thread{1} & \thread{2}\\ \hline
\eventE{1}  & \forkE{\thread{2}}&\\
\eventE{2}  & \writeE{a}&\\
\eventE{3}  & \lockE{x}&\\
\eventE{4}  & \writeE{a}&\\
\eventE{5}  & \unlockE{x}&\\
\eventE{6}  & &\lockE{x}\\
\eventE{7}  & &\writeE{a}\\
\eventE{8}  & &\unlockE{x}\\

 \hline \ea{}

& &

\ba{l}
\mbox{Reordering exhibiting race}

\\

\ba{|l|l|l|}
\hline  & \thread{1} & \thread{2}\\ \hline
\eventE{1}  & \forkE{\thread{2}}&\\
\eventE{6}  & &\lockE{x}\\
\eventE{7}  & &\writeE{a}\\
\eventE{2}  & \writeE{a}&\\

\hline \ea{}

\ea

  \eda

    \caption{Incompleteness due to tracking most recent reads/writes only}
\label{fig:exPWR_2}

\end{figure}

\subsection{PWR Data Race Prediction Enhanced with Approximated Cross-Thread Lock Sets}

\begin{algorithm}
  \caption{PWR Data Race Prediction Enhanced with Cross-Thread Lock Sets} \label{alg:pwr_guards}

\begin{tabular}{cc}

\begin{minipage}{.48\textwidth}

    \bda{ll}
    V & ::= [V_1,...,V_n]  \quad \mbox{Vector clock}
  \\
    V < V'  & \mbox{if}\ \forall k. V_k \leq V'_k \wedge \exists k. V_k < V'_k
  \\
  \pwrVC{e} & \mbox{PWR vector clock}
 \\         & \mbox{for event $e$}
  \\
  \threadVC{t} & \mbox{Thread $t$'s vector clock}
  \\
  \StdLocksSym(t) & \mbox{Thread $t$'s lock set}
  \\
  \AllLocksSym & \mbox{Global lock set}
  \\
  \acqVC{x}&  \mbox{Vector clock of most recent $\lockE{x}$}
  \\
  \RW{a} & \mbox{Race candidates for variable $a$}
  \eda

\begin{algorithmic}[1]
  \Procedure{\mbox{$e$}@acquire}{$t,x$}
  \State $\threadVC{t} = \pwrVC{e}$
\State $\StdLocksSym(t) = \StdLocksSym(t) \cup \{x\}$
\State $\AllLocksSym = \AllLocksSym \cup \{x\}$
\State $\acqVC{x} = \threadVC{t}$
\EndProcedure
\end{algorithmic}

\begin{algorithmic}[1]
  \Procedure{\mbox{$e$}@read}{$t,a$}
  \State $\threadVC{t} = \pwrVC{e}$
    \State $gs = \{ z? \mid z \in \AllLocksSym \setminus \StdLocksSym(t)$
  \\ \quad \quad \quad \quad \quad \quad $\wedge \
               \acqVC{z} < \threadVC{t} \}$
\State $\RW{a} = \RW{a} \cup \{ (e,\threadVC{t}, \StdLocksSym(t),gs) \}$
\EndProcedure
\end{algorithmic}

\begin{algorithmic}[1]
  \Procedure{\mbox{$e$}@write}{$t,a$}
  \State $\threadVC{t} = \pwrVC{e}$
    \State $gs = \{ z? \mid z \in \AllLocksSym \setminus \StdLocksSym(t)$
  \\ \quad \quad \quad \quad \quad \quad $\wedge \
               \acqVC{z} < \threadVC{t} \}$
\State $\RW{a} = \RW{a} \cup \{ (e,\threadVC{t}, \StdLocksSym(t),gs) \}$
\EndProcedure
\end{algorithmic}

\end{minipage}

&

\begin{minipage}{.52\textwidth}

\begin{algorithmic}[1]
  \Procedure{racecheck}{$e,f,a$}
  \If {\begin{tabular}{l}
         $(e,V_e,ls_e, gs_e) \in \RW{a} \wedge \not \exists z? \in gs_e \wedge$
      \\ $(f, V_f, ls_f, gs_f) \in \RW{a} \wedge \not \exists z? \in gs_f \wedge$
      \\ $V_e \not < V_f \wedge V_f \not < V_e \wedge$
      \\ $(ls_e \cup gs_e) \cap (ls_f \cup gs_f) = \emptyset$
        \end{tabular}
  }
  Race found
  \EndIf
\EndProcedure
\end{algorithmic}

\begin{algorithmic}[1]
  \Procedure{\mbox{$e$}@release}{$t,z$}
  \State $\threadVC{t} = \pwrVC{e}$
  \State $\StdLocksSym(t) = \StdLocksSym(t) \setminus \{z\}$
  \State $\AllLocksSym = \AllLocksSym \setminus \{z\}$
  \For {$a \in A$}
  \For {$(e,V,ls, gs) \in \RW{a}$}
  \If {$z? \in gs \wedge V < \threadVC{t}$}
  \State $gs' = (gs \setminus \{ z? \}) \cup \{ z \}$  
  \Else
  \State $gs' = gs \setminus \{ z? \}$ 
  \EndIf
  \State $\RW{a}' = \RW{a} \setminus \{ (e,V,ls, gs) \}$
  \State $\RW{a} = \RW{a}' \cup \{ (e,V,ls, gs') \}$
  \EndFor
  \EndFor
\EndProcedure
\end{algorithmic}

\end{minipage}

\end{tabular}

\end{algorithm}

This section explains how to integrate the computation of
$\PWRLH{\cdot}$ into the PWR data race
predictor~\cite{10.1145/3426182.3426185}.
We do apply the thread-indexed construction covered in section~\ref{sec:elim-more-false}, but omit the details for brevity.
The original PWR data race predictor efficiently implements the PWR relation using vector clocks.
For readability, we abstract these details in a function $\pwrVC{e}$
that computes the vector clock that captures the PWR relation.
That is, $\ltPWRNoT{e}{f}$ iff $\pwrVC{e} < \pwrVC{f}$.
Details on computing $\pwrVC{e}$ may be found elsewhere~\cite{10.1145/3426182.3426185}.

The enhanced PWR algorithm~\ref{alg:pwr_guards} makes use of several state variables.
Each thread maintains its own vector clock $\threadVC{t}$
and (standard) lock set $\StdLocksSym(t)$.
We also maintain a global lock set $\AllLocksSym$.
In  $\acqVC{x}$, we record the vector clock of
the most recent acquire operation on lock $x$.
Race candidates for shared variable $a$ are collected in the set $\RW{a}$.

Events are processed in the customary stream-based fashion.
Each event starts a respective procedure, which takes as an argument a thread id $t$ and either
a variable $a$ or lock $x$.
We use the notation $e@\mbox{operation}$ to uniquely identify the event $e$ associated
with the operation.
For brevity, we ignore fork/join events as their treatment is standard.

Consider processing an acquire operation on $x$ in thread $t$.
We first carry out $\pwrVC{e}$ to compute the PWR vector clock of $e$.
Then, we add $x$ to the thread-local lock set $\StdLocksSym(t)$
and the global lock set $\AllLocksSym$.
Finally, we record the vector clock of the acquire operation.

Processing a read/write event adds a race candidate to the set $\RW{a}$.
Each race candidate is represented as a quadruple $(e, V, ls, gs)$
where $V$ is the vector clock of event $e$,
$ls$ is $e$'s standard lock set and
$gs$ is the set of potential guard locks acquired in threads different from
$\threadID e$.

A potential guard lock is an acquired lock that is not part of the standard
lock set and the acquire operation on this lock took place before the read/write.
See the conditions in line 3 and 4 of \textsc{read} and \textsc{write} procedures.
We mark such potential guard locks with $z?$.
We verify if $z?$ is an actual guard lock
when processing the corresponding release event.

Consider processing a release operation on $z$ in thread $t$.
We deliberately use $z$ so the naming is consistent with the
potential guard locks $z?$ identified in the
\textsc{acquire} procedure.
We consider all race candidates associated with $z?$.
If the race candidate happens before the release (cf.\ line 7), then $z$ is an actual guard lock and part of the set $\PWRLH{e}$.
Otherwise, we remove $z?$.

Our actual implementation does not iterate over all variables $a \in A$ over all race candidates.
Rather, we register race candidates with guard locks $z?$ that need to be checked.
Once the potential guard locks of a race candidate have been checked,
we check if there is an actual race in procedure \textsc{racecheck}.

At least one of the race candidates must be a write.
Candidates are removed aggressively to conserve space. For example,
we maintain at most one write per variable, a read that takes place later evicts
an earlier read etc.
For brevity, such details are ignored in the description of the algorithm.

\mbox{}
\\
\noindent
{\bf Difference to original PWR algorithm.}
The original PWR algorithm can be obtained from
algorithm~\ref{alg:pwr_guards} as follows.
The global lock set~$\AllLocksSym$ can be omitted.
Race candidates are represented as a triple~$(e, V, ls)$.
In procedure \textsc{acquire}, remove lines 3-5.
In procedure \textsc{read}, remove lines 3-4.
In procedure \textsc{write}, remove lines 3-4.
In procedure \textsc{release}, remove lines 4-15.


\section{Experiments}
\label{sec:experiments}

\begin{figure}

  \bda{lrrrrr||rr}
\mbox{Benchmark} & {\mathcal E}  & {\mathcal T}  & {\mathcal M}  & {\mathcal L}  & \mbox{parsing} & \mbox{\pwrCand} & \mbox{\pwrCandGuarded} \\
 \hline
 \mbox{account } & 134 &  7  & 41 & 3 & 2\mbox{ms} & 2\mbox{ms} & 1\mbox{ms} \\
\mbox{airlinetickets } & 140 &  14  & 44 & 0 & 1\mbox{ms} & 1\mbox{ms} & 0\mbox{ms} \\
\mbox{array } & 51 &  6  & 30 & 2 & 0\mbox{ms} & 0\mbox{ms} & 0\mbox{ms} \\
\mbox{batik } & 157.9\mbox{M} &  8  & 4.9\mbox{M} & 1.9\mbox{K} & 7.4\mbox{m} & 12.2\mbox{m} & 13.2\mbox{m} \\
\mbox{boundedbuffer } & 332 &  7  & 63 & 2 & 1\mbox{ms} & 1\mbox{ms} & 1\mbox{ms} \\
\mbox{bubblesort} & 4.6\mbox{K} &  30 & 3 & 168 & 3\mbox{ms} & 37\mbox{ms} & 27\mbox{ms} \\
\mbox{bufwriter } & 22.2\mbox{K} &  9  & 471 & 1 & 51\mbox{ms} & 95\mbox{ms} & 104\mbox{ms} \\
\mbox{clean } & 1.3\mbox{K} &  12  & 26 & 2 & 3\mbox{ms} & 6\mbox{ms} & 7\mbox{ms} \\
\mbox{critical } & 59 &  7  & 30 & 0 & 0\mbox{ms} & 0\mbox{ms} & 0\mbox{ms} \\
\mbox{cryptorsa } & 58.5\mbox{M} &  19  & 1.7\mbox{M} & 8.0\mbox{K} & 2.3\mbox{m} & 3.9\mbox{m} & 4.5\mbox{m} \\
\mbox{derby } & 1.4\mbox{M} &  7  & 185.6\mbox{K} & 1.1\mbox{K} & 3.4\mbox{s} & 6.1\mbox{s} & 8.2\mbox{s} \\
\mbox{ftpserver } & 49.6\mbox{K} &  14  & 5.5\mbox{K} & 301 & 119\mbox{ms} & 232\mbox{ms} & 276\mbox{ms} \\
\mbox{jigsaw } & 3.1\mbox{M} &  15  & 103.5\mbox{K} & 275 & 7.5\mbox{s} & 12.9\mbox{s} & 15.7\mbox{s} \\
\mbox{lang } & 6.3\mbox{K} &  10  & 1.5\mbox{K} & 0 & 16\mbox{ms} & 27\mbox{ms} & 29\mbox{ms} \\
\mbox{linkedlist } & 1.0\mbox{M} &  15  & 3.1\mbox{K} & 1.0\mbox{K} & 2.3\mbox{s} & 4.5\mbox{s} & 4.8\mbox{s} \\
\mbox{lufact } & 134.1\mbox{M} &  5  & 252.1\mbox{K} & 1 & 4.6\mbox{m} & 9.2\mbox{m} & 10.0\mbox{m} \\
\mbox{luindex } & 397.8\mbox{M} &  4  & 2.5\mbox{M} & 65 & 26.1\mbox{m} & 35.7\mbox{m} & 38.5\mbox{m} \\
\mbox{lusearch } & 217.5\mbox{M} &  10  & 5.2\mbox{M} & 118 & 8.6\mbox{m} & 15.1\mbox{m} & 16.3\mbox{m} \\
\mbox{mergesort } & 3.0\mbox{K} &  8  & 621 & 3 & 8\mbox{ms} & 13\mbox{ms} & 14\mbox{ms} \\
\mbox{moldyn } & 200.3\mbox{K} &  6  & 1.2\mbox{K} & 2 & 457\mbox{ms} & 754\mbox{ms} & 807\mbox{ms} \\
\mbox{pingpong } & 151 &  21  & 51 & 0 & 1\mbox{ms} & 0\mbox{ms} & 1\mbox{ms} \\
\mbox{producerconsumer } & 658 &  11  & 67 & 3 & 1\mbox{ms} & 3\mbox{ms} & 3\mbox{ms} \\
\mbox{raytracer } & 15.8\mbox{K} &  6  & 3.9\mbox{K} & 8 & 37\mbox{ms} & 65\mbox{ms} & 70\mbox{ms} \\
\mbox{readerswriters } & 11.3\mbox{K} &  8  & 18 & 1 & 26\mbox{ms} & 54\mbox{ms} & 60\mbox{ms} \\
\mbox{sor } & 606.9\mbox{M} &  5  & 1.0\mbox{M} & 2 & 23.1\mbox{m} & 37.1\mbox{m} & 40.4\mbox{m} \\
\mbox{sunflow } & 11.7\mbox{M} &  17  & 1.3\mbox{M} & 9 & 27.5\mbox{s} & 55.2\mbox{s} & 59.7\mbox{s} \\
\mbox{tsp } & 307.3\mbox{M} &  10  & 181.1\mbox{K} & 2 & 13.6\mbox{m} & 23.6\mbox{m} & 25.5\mbox{m} \\
\mbox{twostage } & 193 &  15  & 21 & 2 & 0\mbox{ms} & 1\mbox{ms} & 1\mbox{ms} \\
\mbox{wronglock } & 246 &  25  & 26 & 2 & 0\mbox{ms} & 1\mbox{ms} & 1\mbox{ms} \\
\mbox{xalan } & 122.5\mbox{M} &  9  & 4.4\mbox{M} & 2.5\mbox{K} & 5.1\mbox{m} & 8.9\mbox{m} & 24.4\mbox{m} \\
\mbox{biojava } & 221.0\mbox{M} &  6  & 121.0\mbox{K} & 78 & 31.5\mbox{m} & 14.4\mbox{m} & 15.5\mbox{m} \\
\mbox{cassandra } & 259.1\mbox{M} &  178  & 9.9\mbox{M} & 60.5\mbox{K} & 11.0\mbox{m} & 22.6\mbox{m} & 25.0\mbox{m} \\
\mbox{graphchi } & 215.8\mbox{M} &  20  & 24.9\mbox{M} & 60 & 9.1\mbox{m} & 15.7\mbox{m} & 17.0\mbox{m} \\
\mbox{hsqldb } & 18.8\mbox{M} &  46  & 945.0\mbox{K} & 401 & 47.6\mbox{s} & 1.3\mbox{m} & 1.5\mbox{m} \\
\mbox{tradebeans } & 39.1\mbox{M} &  236  & 2.8\mbox{M} & 6.1\mbox{K} & 1.6\mbox{m} & 3.5\mbox{m} & 3.9\mbox{m} \\
\mbox{tradesoap } & 39.1\mbox{M} &  236  & 2.8\mbox{M} & 6.1\mbox{K} & 1.6\mbox{m} & 3.4\mbox{m} & 3.9\mbox{m} \\
\mbox{zxing } & 546.4\mbox{M} &  15  & 37.8\mbox{M} & 1.5\mbox{K} & 24.0\mbox{m} & 46.6\mbox{m} & 49.2\mbox{m} \\

  \hline \eda
  \caption{{\bf Running times}.
    Columns ${\mathcal E}$, ${\mathcal T}$, ${\mathcal M}$ and ${\mathcal L}$ contain
    the number of events, total number of threads, total number of memory locations and total number of
    locks.
    Column ``parsing'' contains the running time to process all events.
    Columns \pwrCand\ and \pwrCandGuarded\ contain
    the running times of our two test candidates.}
  \label{fig:results-misc-zero}
\end{figure}

\begin{figure}

\bda{l||r|r|r}
\mbox{Benchmark} & \mbox{races} & \mbox{cs guarded} & \mbox{cs guard locks} \\
 \hline
 \mbox{account } & 1 &  0  & 0 \\
\mbox{airlinetickets } & 2 &  0  & 0 \\
\mbox{array } & 1 &  0  & 0 \\
\mbox{batik } & 1 &  0  & 0 \\
\mbox{boundedbuffer } & 1 &  0  & 0 \\
\mbox{bubblesort } & 3 &  0  & 0 \\
\mbox{bufwriter } & 2 &  0  & 0 \\
\mbox{clean } & 3 &  0  & 0 \\
\mbox{critical } & 3 &  0  & 0 \\
\mbox{cryptorsa } & 4 &  0  & 0 \\
\mbox{derby } & 10 &  2  & 109.0\mbox{K} \\
\mbox{ftpserver } & 22 &  2  & 1.3\mbox{K} \\
\mbox{jigsaw } & 6 &  0  & 0 \\
\mbox{lang } & 1 &  0  & 0 \\
\mbox{linkedlist } & 4 &  0  & 0 \\
\mbox{lufact } & 2 &  0  & 0 \\
\mbox{luindex } & 10 &  0  & 0 \\
\mbox{lusearch } & 42 &  0  & 0 \\
\mbox{mergesort } & 1 &  0  & 0 \\
\mbox{moldyn } & 2 &  0  & 0 \\
\mbox{pingpong } & 2 &  0  & 0 \\
\mbox{producerconsumer } & 1 &  0  & 0 \\
\mbox{raytracer } & 3 &  0  & 0 \\
\mbox{readerswriters } & 4 &  0  & 0 \\
\mbox{sor } & 0 &  0  & 0 \\
\mbox{sunflow } & 5 &  0  & 0 \\
\mbox{tsp } & 5 &  0  & 0 \\
\mbox{twostage } & 1 &  0  & 0 \\
\mbox{wronglock } & 2 &  0  & 0 \\
\mbox{xalan } & 11 &  0  & 25.1\mbox{M} \\
\mbox{biojava } & 1 &  0  & 0 \\
\mbox{cassandra } & 48 &  1  & 767.2\mbox{K} \\
\mbox{graphchi } & 4 &  0  & 633.2\mbox{K} \\
\mbox{hsqldb } & 180 &  0  & 0 \\
\mbox{tradebeans } & 29 &  1  & 1.0\mbox{K} \\
\mbox{tradesoap } & 28 &  0  & 991 \\
\mbox{zxing } & 61 &  0  & 0 \\

  \hline \eda

  \caption{{\bf Comparison data races}.
    Column ``races'' contains the number of distinct source code
    locations where \pwrCand\ predicted a race.
    Column ``cs guarded'' contains the number of source code locations
    for which \pwrCandGuarded\ finds
    a guard lock that is not in the standard lock set but arises from a cross-thread lock set.
    Column ``cs guard locks'' contains the overall number of times \pwrCandGuarded\ encountered a cross-thread (guard) lock}
  \label{fig:results-races-all}
\end{figure}

We implemented the original PWR data race predictor as well as the
enhanced version as discussed in section~\ref{alg:pwr_guards} in C++.
In this section, we use the implementations to answer the research
questions posed in the overview:

\begin{description}
\item[RQ1] Can we compute cross-thread critical sections and the associated
  lock sets efficiently?

\item[RQ2] What is the effect on the analysis results of  incorporating
  cross-thread critical section?
\end{description}

\smallskip
\noindent
{\bf Test candidates.}
\pwrCand\ is our port of the original Go implementation of PWR~\cite{speedy-go} to C++.
\pwrCandGuarded\ is our enhancement of PWR described in algorithm~\ref{alg:pwr_guards}.
All test candidates are available in a GitHub
repository.\footnote{
  GitHub URL withheld to maintain anonymity.
}

\smallskip
\noindent
{\bf Benchmarks and system setup.}
Our experiments cover all traces of the zero-reversal-logs and misc collection of traces that are provided by~\citet{tracelogs}.
These traces are obtained from standard benchmark suites such as
Java Grande~\cite{smith2001parallel}, Da
Capo~\cite{Blackburn:2006:DBJ:1167473.1167488}, and IBM
Contest~\cite{farchi2003concurrent}.
It is unknown how many "real" data races are hidden in these trace logs.
We share this issue with prior work that uses these trace logs for evaluation.
We conducted our experiments on an Apple M1 max CPU with 32GB of RAM running macOS Monterey (Version 12.1).

\smallskip
\noindent
{\bf Performance.}
Figure~\ref{fig:results-misc-zero} shows the detailed results of our performance measurements.
The overall overhead to compute cross-thread lock sets for the purpose of data race prediction is manageable
in all cases.
When comparing \pwrCandGuarded\ against \pwrCand\ the worst case
slowdown is 2.7x in case of xalan, but in most cases the slowdown is
between 10-20\%.
Hence, we can answer {\bf RQ1} affirmatively.

\smallskip
\noindent
{\bf Data races and guard locks.}
Figure~\ref{fig:results-races-all} shows the impact
of cross-thread lock sets on data race prediction.

For ftpserver, \pwrCand\ finds  22 distinct source code locations that triggered a race.
Recall that \pwrCand\ uses the standard lock set to reduce the number of race candidates.
Out of the 22 source code locations there are two locations that are additionally guarded by cross-thread locks predicted
by \pwrCandGuarded.
Hence, the programmer should focus first on the 20 racy source code locations that are not guarded by cross-thread locks.
For the cases shown in Figure~\ref{fig:results-races-all}, we can thus eliminate up to 20\% (derby) of the race candidates
that need to be inspected by the user.
This positively answers~{\bf RQ2}.

For xalan, we find 11 race locations but none is guarded by a cross-thread lock.
However, \pwrCandGuarded  encounters a guard lock 25  million times,
which is the number of times condition $V < \threadVC{t}$ on line 7 in procedure \textsc{release} in
algorithm~\ref{alg:pwr_guards} is satisfied.
Such huge numbers likely arise because of the high number of locks and
memory locations (variables) in the xalan trace.


Overall, there are 37 trace logs for which we report 500 racy source code locations
out of which 6 false positives can be eliminated  by \pwrCandGuarded.
The improvement is not drastic but it shows that we can further advance
the state of the art when it comes to efficient, near complete but often sound dynamic race prediction.
Similar observations apply to other state of the art data race prediction tools.
To quote the WCP paper:
\begin{quote}
  ``Note that, our WCP based race detection algorithm does
     not report drastically more races than the simpler HB based
     algorithm.''
\end{quote}
Importantly, we can show that cross-thread critical sections arise in practical programs.

\smallskip
\noindent
{\bf Comparison of PWR against WCP, FastTrack and others.}
We measure here only the relative improvements for PWR by taking into account
cross-thread critical sections.
A detailed comparison of PWR against WCP, FastTrack and others is given in~\cite{10.1145/3426182.3426185,DBLP:phd/dnb/Stadtmuller21}.


\section{Related Work}
\label{sec:related}

We started sections~\ref{sec:introduction}
and~\ref{sec:overview} with a review of selected work on
dynamic data race prediction. Here we comment on some further
work divided in four categories. Remember, when we speak of lock set
here, we mean the standard version. Trace analysis with lock sets and
partial orders is also applicable to deadlock detection, but we
concentrate on the analysis of data races in this work.

\medskip
\noindent
{\bf Lock set-based methods.}
\citet{Dinning:1991:DAA:127695:122767}
introduce the idea of a lock set for the purpose of dynamic data race
analysis as described in section~\ref{sec:overview}.
Subsequent work \cite{elmas2010goldilocks,
  Savage:1997:EDD:265924.265927} refines this approach, as the pure
lock set method is fast and complete, but inherently prone to false
positives. These false positives are often addressed using
partial-order-based methods.

\medskip
\noindent
{\bf Combinatorial methods.}
These methods are sound \emph{and} complete, as they attempt to cover
all correct reorderings exhaustively. They often employ SAT/SMT solvers to
derive alternative schedules from the constraints extracted from the trace.
Their run time can be exponential in the worst case.

\citet{sen2005detecting} define a framework for predictive runtime
analysis which enables one to systematically explore properties in all
equivalence classes of reorderings of a trace. As an example, they
encode datarace analysis in their framework. An implementation is
reported, but no quantitative data about it is available.

\citet{Chen:2007:PSC:1770351.1770387} offer a parametric framework
for causality to prove properties of a range of
happen-before relations. There is an implemented instance of the
framework for a particular algorithm based on vector clocks. They find
additional dataraces in small applications.

\citet{Said:2011:GDR:1986308.1986334} design a symbolic analysis based
on SMT solving to search for alternative traces witnessing a data
race. Their analysis encodes the sequential consistency semantics and
is said to improve over the maximal causal model. They present
preliminary statistics of an implementation on small benchmarks. Their
algorithm reduces the number of false positives significantly.

\citet{DBLP:conf/rv/SerbanutaCR12} investigate maximal causal models,
which determine all alternative traces that all programs exhibiting a
given trace might execute. They conduct a semantic investigation, but
there is no implementation.

\citet{Huang:2014:MSP:2666356.2594315} encode control flow into
constraints and formulate race detection as a constraint
problem. They prove an optimality result and claim scalability to
realistic applications.

\citet{LuoHuangRosuSystematic15} propose a constraint-based approach
to generate a sufficient number of alternative schedules to cover the
maximal causal space. Their goal is to discover concurrency bugs,
which are not further specified.

\medskip
\noindent
{\bf Partial-order-based methods.}
These methods are efficient as they infer a partial order via a linear pass through the trace.
We already discussed
WCP~\cite{Kini:2017:DRP:3140587.3062374}, SDP and WDP~\cite{10.1145/3360605},
DC~\cite{Roemer:2018:HUS:3296979.3192385}, and {\DCTaskPar}
\cite{8894270} in section~\ref{sec:overview}. There is a separate
section~\ref{sec:pwr} with details on PWR.

SHB~\citep{Mathur:2018:HFR:3288538:3276515} strengthens
HB~\cite{lamport1978time} by guaranteeing that all detected
races have viable schedules (HB only guarantees that for the first
detected race).
Like HB, SHB maintains the textual order of critical sections.
As each acquire synchronizes with the prior release,
HB and SHB are not affected by cross-thread critical sections.

The Goldilocks algorithm~\cite{10.1007/11940197_13}. 
computes the HB relation based on a combination of lock sets and thread ids.
The textual order of critical sections is maintained in this construction.
Hence, Goldilocks is not affected by cross-thread critical sections.

Many of these methods have efficient implementations.
For example, consider FastTrack~\citep{flanagan2010fasttrack} which implements HB
and SmartTrack~\cite{10.1145/3385412.3385993} which further optimizes DC.

SyncP~\cite{10.1145/3434317} maintains the textual order among critical sections (like HB)
but allows to skip certain parts of a trace.
For example, SyncP is able to predict the race in figure~\ref{fig:exPWR_2}.
Such races are referred to as sync-preserving races.
Strictly speaking, SyncP is not a partial-order method but we place it here because of its connection to HB.
Like HB, SyncP is not affected by cross-thread critical sections.




We already discussed TSan~\citep{serebryany2009threadsanitizer} in the introduction.
The newer TSan version, ThreadSanitizer v2
(TSanV2)~\citep{thread-sanitizer-v2}, keeps a limited history of
write/read events to  improve the performance with respect to
FastTrack, which leads to a higher number of false negatives.

Many algorithms combine partial orders with lock sets.
Acculock \citep{xie2013acculock}  optimizes the original TSan algorithm
by employing a single lockset per variable.
Acculock can be faster, but is less precise compared to TSan
if a thread uses multiple locks at once.

SimpleLock \citep{yu2016simplelock+} simplifies the lockset algorithm.
A data race is only reported if at least one of the accesses is not protected by any lock.
SimpleLock is faster compared to Acculock but misses more data races
since it ignores data races for events with different locks.


\citet{10.1145/781498.781528} compute
a weak form of happens-before relation  in addition to lock sets.
With such a hybrid method,
conflicting events with disjoint lock sets are ignored if they can be ordered with their happens-before relation.
Thus, some false positives can be eliminated.
Similar ideas can be found in the work of
\citet{10.1145/3360605}
as well as
\citet{10.1145/3426182.3426185}.




\medskip
\noindent
{\bf Multi-phase methods.}
The Vindicator algorithm~\citep{Roemer:2018:HUS:3296979.3192385} improves the WCP algorithm
and is sound for all reported data races. It can predict more data races compared to WCP,
but requires three phases to do so. The first phase of Vindicator is selecting
races based on the DC relation (a weakened WCP relation that removes the happens-before closure).
For the second phase, it constructs a graph that
contains all events from the processed trace. This phase is unsound and incomplete
which is why a third phase is required. The third phase makes a single attempt
to reconstruct a witness trace for the potential data race and reports a data race if successful.

The M2 algorithm~\citep{Andreas:Pavlogiannis:popl20} can be seen as a further improvement of the Vindicator idea.
Like Vindicator, multiple phases are required. M2 requires two phases.
M2 has $O(n^4)$ run-time (where $n$ is the size of the trace).
M2 is sound and complete for two threads.

Vindicator and M2 rely on a trace reordering phase.
Because this phase obeys the laws of being a well-formed, cross-thread critical sections have no impact.

\section{Conclusion}
\label{sec:Conclusion}

Starting from our observation ``in the wild'' that not all critical sections are
confined to a single thread, we developed a theoretical framework of
critical sections that may extend across more than one thread and the
accompanying lock set constructions. Using the framework, we showed
how the lock set construction directly impacts false positives and
false negatives. Subsequently, we use the framework to analyse an
approximate lock set construction, which is based on the partial order
PWR.

From this basis, we considered the state-of-the-art, efficient partial
order-based methods for data race detection and studied the impact of
cross-thread critical sections on them.
WCP and SDP are agnostic about cross-thread critical sections.
For WCP we can show that the WCP soundness proof extends to the case of cross-thread critical sections.
For SDP the details have yet to be worked out.
We diagnosed serious soundness
problems for {\DCTaskPar}, which we leave to fix in future
work. The complete relations WDP and PWR were fixable and we studied
the efficiency and the effectiveness of an implementation that takes
cross-thread critical sections into account. This implementation is
based on PWR.

To reiterate again: if any of the related works in
section~\ref{sec:related} makes use of lock sets, they employ the
standard notion, where the lock set of an event only contains locks acquired in
the same thread as the event. The technical results in the references we
studied in detail depend in various degrees on this assumption on the
events in a lock set. It seems wise to scrutinize other results
with cross-thread lock sets in mind, too.
We also expect some impact on algorithms for deadlock detection, but
leave that investigation to future work.


\begin{acks}

We thank some POPL'24 referees for their comments on a previous version of this paper.

\end{acks}


\bibliography{../bib/main}


\begin{thebibliography}{37}


\ifx \showCODEN    \undefined \def \showCODEN     #1{\unskip}     \fi
\ifx \showDOI      \undefined \def \showDOI       #1{#1}\fi
\ifx \showISBNx    \undefined \def \showISBNx     #1{\unskip}     \fi
\ifx \showISBNxiii \undefined \def \showISBNxiii  #1{\unskip}     \fi
\ifx \showISSN     \undefined \def \showISSN      #1{\unskip}     \fi
\ifx \showLCCN     \undefined \def \showLCCN      #1{\unskip}     \fi
\ifx \shownote     \undefined \def \shownote      #1{#1}          \fi
\ifx \showarticletitle \undefined \def \showarticletitle #1{#1}   \fi
\ifx \showURL      \undefined \def \showURL       {\relax}        \fi
\providecommand\bibfield[2]{#2}
\providecommand\bibinfo[2]{#2}
\providecommand\natexlab[1]{#1}
\providecommand\showeprint[2][]{arXiv:#2}

\bibitem[\protect\citeauthoryear{Adve and Gharachorloo}{Adve and
  Gharachorloo}{1996}]%
        {Adve:1996:SMC:619013.620590}
\bibfield{author}{\bibinfo{person}{Sarita~V. Adve} {and}
  \bibinfo{person}{Kourosh Gharachorloo}.} \bibinfo{year}{1996}\natexlab{}.
\newblock \showarticletitle{Shared Memory Consistency Models: A Tutorial}.
\newblock \bibinfo{journal}{\emph{Computer}} \bibinfo{volume}{29},
  \bibinfo{number}{12} (\bibinfo{date}{Dec.} \bibinfo{year}{1996}),
  \bibinfo{pages}{66--76}.
\newblock
\showISSN{0018-9162}
\urldef\tempurl%
\url{https://doi.org/10.1109/2.546611}
\showDOI{\tempurl}


\bibitem[\protect\citeauthoryear{Blackburn, Garner, Hoffmann, Khang, McKinley,
  Bentzur, Diwan, Feinberg, Frampton, Guyer, Hirzel, Hosking, Jump, Lee, Moss,
  Phansalkar, Stefanovi\'{c}, VanDrunen, von Dincklage, and
  Wiedermann}{Blackburn et~al\mbox{.}}{2006}]%
        {Blackburn:2006:DBJ:1167473.1167488}
\bibfield{author}{\bibinfo{person}{Stephen~M. Blackburn},
  \bibinfo{person}{Robin Garner}, \bibinfo{person}{Chris Hoffmann},
  \bibinfo{person}{Asjad~M. Khang}, \bibinfo{person}{Kathryn~S. McKinley},
  \bibinfo{person}{Rotem Bentzur}, \bibinfo{person}{Amer Diwan},
  \bibinfo{person}{Daniel Feinberg}, \bibinfo{person}{Daniel Frampton},
  \bibinfo{person}{Samuel~Z. Guyer}, \bibinfo{person}{Martin Hirzel},
  \bibinfo{person}{Antony Hosking}, \bibinfo{person}{Maria Jump},
  \bibinfo{person}{Han Lee}, \bibinfo{person}{J.~Eliot~B. Moss},
  \bibinfo{person}{Aashish Phansalkar}, \bibinfo{person}{Darko Stefanovi\'{c}},
  \bibinfo{person}{Thomas VanDrunen}, \bibinfo{person}{Daniel von Dincklage},
  {and} \bibinfo{person}{Ben Wiedermann}.} \bibinfo{year}{2006}\natexlab{}.
\newblock \showarticletitle{The {DaCapo} Benchmarks: {Java} Benchmarking
  Development and Analysis}. In \bibinfo{booktitle}{\emph{Proc.\ of OOPSLA
  '06}}. \bibinfo{publisher}{ACM}, \bibinfo{address}{Portland, OR, USA},
  \bibinfo{pages}{169--190}.
\newblock
\urldef\tempurl%
\url{https://doi.org/10.1145/1167515.1167488}
\showDOI{\tempurl}


\bibitem[\protect\citeauthoryear{Chen and Ro\c{s}u}{Chen and Ro\c{s}u}{2007}]%
        {Chen:2007:PSC:1770351.1770387}
\bibfield{author}{\bibinfo{person}{Feng Chen} {and} \bibinfo{person}{Grigore
  Ro\c{s}u}.} \bibinfo{year}{2007}\natexlab{}.
\newblock \showarticletitle{Parametric and Sliced Causality}. In
  \bibinfo{booktitle}{\emph{Proc.\ of CAV'07}}. \bibinfo{publisher}{Springer},
  \bibinfo{address}{Berlin, Germany}, \bibinfo{pages}{240--253}.
\newblock
\showISBNx{978-3-540-73367-6}
\urldef\tempurl%
\url{http://dl.acm.org/citation.cfm?id=1770351.1770387}
\showURL{%
\tempurl}


\bibitem[\protect\citeauthoryear{Dijkstra}{Dijkstra}{1965}]%
        {DBLP:journals/cacm/Dijkstra65}
\bibfield{author}{\bibinfo{person}{Edsger~W. Dijkstra}.}
  \bibinfo{year}{1965}\natexlab{}.
\newblock \showarticletitle{Solution of a problem in concurrent programming
  control}.
\newblock \bibinfo{journal}{\emph{Commun. {ACM}}} \bibinfo{volume}{8},
  \bibinfo{number}{9} (\bibinfo{year}{1965}), \bibinfo{pages}{569}.
\newblock
\urldef\tempurl%
\url{https://doi.org/10.1145/365559.365617}
\showDOI{\tempurl}


\bibitem[\protect\citeauthoryear{Dinning and Schonberg}{Dinning and
  Schonberg}{1991}]%
        {Dinning:1991:DAA:127695:122767}
\bibfield{author}{\bibinfo{person}{Anne Dinning} {and} \bibinfo{person}{Edith
  Schonberg}.} \bibinfo{year}{1991}\natexlab{}.
\newblock \showarticletitle{Detecting Access Anomalies in Programs with
  Critical Sections}.
\newblock \bibinfo{journal}{\emph{SIGPLAN Not.}} \bibinfo{volume}{26},
  \bibinfo{number}{12} (\bibinfo{date}{Dec.} \bibinfo{year}{1991}),
  \bibinfo{pages}{85--96}.
\newblock
\showISSN{0362-1340}
\urldef\tempurl%
\url{https://doi.org/10.1145/127695.122767}
\showDOI{\tempurl}


\bibitem[\protect\citeauthoryear{Elmas, Qadeer, and Tasiran}{Elmas
  et~al\mbox{.}}{2006}]%
        {10.1007/11940197_13}
\bibfield{author}{\bibinfo{person}{Tayfun Elmas}, \bibinfo{person}{Shaz
  Qadeer}, {and} \bibinfo{person}{Serdar Tasiran}.}
  \bibinfo{year}{2006}\natexlab{}.
\newblock \showarticletitle{Goldilocks: Efficiently Computing the
  Happens-before Relation Using Locksets}. In
  \bibinfo{booktitle}{\emph{Proceedings of the First Combined International
  Conference on Formal Approaches to Software Testing and Runtime
  Verification}} (Seattle, WA) \emph{(\bibinfo{series}{FATES'06/RV'06})}.
  \bibinfo{publisher}{Springer-Verlag}, \bibinfo{address}{Berlin, Heidelberg},
  \bibinfo{pages}{193–208}.
\newblock
\showISBNx{3540496998}
\urldef\tempurl%
\url{https://doi.org/10.1007/11940197_13}
\showDOI{\tempurl}


\bibitem[\protect\citeauthoryear{Elmas, Qadeer, and Tasiran}{Elmas
  et~al\mbox{.}}{2010}]%
        {elmas2010goldilocks}
\bibfield{author}{\bibinfo{person}{Tayfun Elmas}, \bibinfo{person}{Shaz
  Qadeer}, {and} \bibinfo{person}{Serdar Tasiran}.}
  \bibinfo{year}{2010}\natexlab{}.
\newblock \showarticletitle{Goldilocks: a race-aware {Java} runtime}.
\newblock \bibinfo{journal}{\emph{Commun. ACM}} \bibinfo{volume}{53},
  \bibinfo{number}{11} (\bibinfo{year}{2010}), \bibinfo{pages}{85--92}.
\newblock


\bibitem[\protect\citeauthoryear{Farchi, Nir, and Ur}{Farchi
  et~al\mbox{.}}{2003}]%
        {farchi2003concurrent}
\bibfield{author}{\bibinfo{person}{Eitan Farchi}, \bibinfo{person}{Yarden Nir},
  {and} \bibinfo{person}{Shmuel Ur}.} \bibinfo{year}{2003}\natexlab{}.
\newblock \showarticletitle{Concurrent Bug Patterns and How to Test Them}. In
  \bibinfo{booktitle}{\emph{Proceedings International Parallel and Distributed
  Processing Symposium}}. \bibinfo{publisher}{IEEE}, \bibinfo{pages}{7--pp}.
\newblock
\urldef\tempurl%
\url{https://doi.org/10.1109/IPDPS.2003.1213511}
\showDOI{\tempurl}


\bibitem[\protect\citeauthoryear{Flanagan and Freund}{Flanagan and
  Freund}{2010}]%
        {flanagan2010fasttrack}
\bibfield{author}{\bibinfo{person}{Cormac Flanagan} {and}
  \bibinfo{person}{Stephen~N. Freund}.} \bibinfo{year}{2010}\natexlab{}.
\newblock \showarticletitle{{FastTrack}: Efficient and Precise Dynamic Race
  Detection}.
\newblock \bibinfo{journal}{\emph{Commun. ACM}} \bibinfo{volume}{53},
  \bibinfo{number}{11} (\bibinfo{year}{2010}), \bibinfo{pages}{93--101}.
\newblock
\urldef\tempurl%
\url{https://doi.org/10.1145/1543135.1542490}
\showDOI{\tempurl}


\bibitem[\protect\citeauthoryear{Gen\c{c}, Roemer, Xu, and Bond}{Gen\c{c}
  et~al\mbox{.}}{2019}]%
        {10.1145/3360605}
\bibfield{author}{\bibinfo{person}{Kaan Gen\c{c}}, \bibinfo{person}{Jake
  Roemer}, \bibinfo{person}{Yufan Xu}, {and} \bibinfo{person}{Michael~D.
  Bond}.} \bibinfo{year}{2019}\natexlab{}.
\newblock \showarticletitle{Dependence-Aware, Unbounded Sound Predictive Race
  Detection}.
\newblock \bibinfo{journal}{\emph{Proc. ACM Program. Lang.}}
  \bibinfo{volume}{3}, \bibinfo{number}{OOPSLA}, Article
  \bibinfo{articleno}{179} (\bibinfo{date}{Oct.} \bibinfo{year}{2019}),
  \bibinfo{numpages}{30}~pages.
\newblock
\urldef\tempurl%
\url{https://doi.org/10.1145/3360605}
\showDOI{\tempurl}


\bibitem[\protect\citeauthoryear{Huang, Meredith, and Ro\c{s}u}{Huang
  et~al\mbox{.}}{2014}]%
        {Huang:2014:MSP:2666356.2594315}
\bibfield{author}{\bibinfo{person}{Jeff Huang}, \bibinfo{person}{Patrick~O'Neil
  Meredith}, {and} \bibinfo{person}{Grigore Ro\c{s}u}.}
  \bibinfo{year}{2014}\natexlab{}.
\newblock \showarticletitle{Maximal Sound Predictive Race Detection With
  Control Flow Abstraction}. In \bibinfo{booktitle}{\emph{{PLDI} '14}},
  \bibfield{editor}{\bibinfo{person}{Michael F.~P. O'Boyle} {and}
  \bibinfo{person}{Keshav Pingali}} (Eds.). \bibinfo{publisher}{{ACM}},
  \bibinfo{address}{Edinburgh, United Kingdom}, \bibinfo{pages}{337--348}.
\newblock
\urldef\tempurl%
\url{https://doi.org/10.1145/2594291.2594315}
\showDOI{\tempurl}


\bibitem[\protect\citeauthoryear{Kini, Mathur, and Viswanathan}{Kini
  et~al\mbox{.}}{2017a}]%
        {Kini:2017:DRP:3140587.3062374}
\bibfield{author}{\bibinfo{person}{Dileep Kini}, \bibinfo{person}{Umang
  Mathur}, {and} \bibinfo{person}{Mahesh Viswanathan}.}
  \bibinfo{year}{2017}\natexlab{a}.
\newblock \showarticletitle{Dynamic Race Prediction in Linear Time}.
\newblock \bibinfo{journal}{\emph{SIGPLAN Not.}} \bibinfo{volume}{52},
  \bibinfo{number}{6} (\bibinfo{date}{June} \bibinfo{year}{2017}),
  \bibinfo{pages}{157--170}.
\newblock
\urldef\tempurl%
\url{https://doi.org/10.1145/3062341.3062374}
\showDOI{\tempurl}


\bibitem[\protect\citeauthoryear{Kini, Mathur, and Viswanathan}{Kini
  et~al\mbox{.}}{2017b}]%
        {DBLP:journals/corr/KiniM017}
\bibfield{author}{\bibinfo{person}{Dileep Kini}, \bibinfo{person}{Umang
  Mathur}, {and} \bibinfo{person}{Mahesh Viswanathan}.}
  \bibinfo{year}{2017}\natexlab{b}.
\newblock \showarticletitle{Dynamic Race Prediction in Linear Time}.
\newblock \bibinfo{journal}{\emph{CoRR}}  \bibinfo{volume}{abs/1704.02432}
  (\bibinfo{year}{2017}), 22.
\newblock
\showeprint[arXiv]{1704.02432}
\urldef\tempurl%
\url{http://arxiv.org/abs/1704.02432}
\showURL{%
\tempurl}


\bibitem[\protect\citeauthoryear{Lamport}{Lamport}{1978}]%
        {lamport1978time}
\bibfield{author}{\bibinfo{person}{Leslie Lamport}.}
  \bibinfo{year}{1978}\natexlab{}.
\newblock \showarticletitle{Time, Clocks, and the Ordering of Events in a
  Distributed System}.
\newblock \bibinfo{journal}{\emph{Commun. ACM}} \bibinfo{volume}{21},
  \bibinfo{number}{7} (\bibinfo{year}{1978}), \bibinfo{pages}{558--565}.
\newblock
\urldef\tempurl%
\url{https://doi.org/10.1145/359545.359563}
\showDOI{\tempurl}


\bibitem[\protect\citeauthoryear{Luo, Huang, and Ro\c{s}u}{Luo
  et~al\mbox{.}}{2015}]%
        {LuoHuangRosuSystematic15}
\bibfield{author}{\bibinfo{person}{Qingzhou Luo}, \bibinfo{person}{Jeff Huang},
  {and} \bibinfo{person}{Grigore Ro\c{s}u}.} \bibinfo{year}{2015}\natexlab{}.
\newblock \bibinfo{booktitle}{\emph{Systematic Concurrency Testing with Maximal
  Causality}}.
\newblock \bibinfo{type}{{T}echnical {R}eport}.
\newblock


\bibitem[\protect\citeauthoryear{Mathur, Kini, and Viswanathan}{Mathur
  et~al\mbox{.}}{2018}]%
        {Mathur:2018:HFR:3288538:3276515}
\bibfield{author}{\bibinfo{person}{Umang Mathur}, \bibinfo{person}{Dileep
  Kini}, {and} \bibinfo{person}{Mahesh Viswanathan}.}
  \bibinfo{year}{2018}\natexlab{}.
\newblock \showarticletitle{What Happens-after the First Race? Enhancing the
  Predictive Power of Happens-before Based Dynamic Race Detection}.
\newblock \bibinfo{journal}{\emph{Proc. ACM Program. Lang.}}
  \bibinfo{volume}{2}, \bibinfo{number}{OOPSLA}, Article
  \bibinfo{articleno}{145} (\bibinfo{date}{Oct.} \bibinfo{year}{2018}),
  \bibinfo{numpages}{29}~pages.
\newblock
\showISSN{2475-1421}
\urldef\tempurl%
\url{https://doi.org/10.1145/3276515}
\showDOI{\tempurl}


\bibitem[\protect\citeauthoryear{Mathur, Pavlogiannis, Tun\c{c}, and
  Viswanathan}{Mathur et~al\mbox{.}}{[n.d.]}]%
        {tracelogs}
\bibfield{author}{\bibinfo{person}{Umang Mathur}, \bibinfo{person}{Andreas
  Pavlogiannis}, \bibinfo{person}{H\"{u}nkar~Can Tun\c{c}}, {and}
  \bibinfo{person}{Mahesh Viswanathan}.} \bibinfo{year}{[n.d.]}\natexlab{}.
\newblock \bibinfo{title}{Trace Logs}.
\newblock
\newblock
\urldef\tempurl%
\url{https://uillinoisedu-my.sharepoint.com/:f:/g/personal/umathur3_illinois_edu/EskC1fg2xhNHnim2ZYjDD9gBJqme8hBTgWShHUmOfYmF-Q?e=QE762I}
\showURL{%
\tempurl}


\bibitem[\protect\citeauthoryear{Mathur, Pavlogiannis, and Viswanathan}{Mathur
  et~al\mbox{.}}{2021}]%
        {10.1145/3434317}
\bibfield{author}{\bibinfo{person}{Umang Mathur}, \bibinfo{person}{Andreas
  Pavlogiannis}, {and} \bibinfo{person}{Mahesh Viswanathan}.}
  \bibinfo{year}{2021}\natexlab{}.
\newblock \showarticletitle{Optimal Prediction of Synchronization-Preserving
  Races}.
\newblock \bibinfo{journal}{\emph{Proc. ACM Program. Lang.}}
  \bibinfo{volume}{5}, \bibinfo{number}{POPL}, Article \bibinfo{articleno}{36}
  (\bibinfo{date}{jan} \bibinfo{year}{2021}), \bibinfo{numpages}{29}~pages.
\newblock
\urldef\tempurl%
\url{https://doi.org/10.1145/3434317}
\showDOI{\tempurl}


\bibitem[\protect\citeauthoryear{Mathur and Tun\c{c}}{Mathur and
  Tun\c{c}}{2020}]%
        {rapid-mathur}
\bibfield{author}{\bibinfo{person}{Umang Mathur} {and}
  \bibinfo{person}{H\"{u}nkar~Can Tun\c{c}}.} \bibinfo{year}{2020}\natexlab{}.
\newblock \bibinfo{title}{{RAPID} : Dynamic Analysis for Concurrent Programs}.
\newblock \bibinfo{howpublished}{https://github.com/umangm/rapid}.
\newblock


\bibitem[\protect\citeauthoryear{O'Callahan and Choi}{O'Callahan and
  Choi}{2003}]%
        {10.1145/781498.781528}
\bibfield{author}{\bibinfo{person}{Robert O'Callahan} {and}
  \bibinfo{person}{Jong-Deok Choi}.} \bibinfo{year}{2003}\natexlab{}.
\newblock \showarticletitle{Hybrid Dynamic Data Race Detection}. In
  \bibinfo{booktitle}{\emph{Proceedings of the Ninth ACM SIGPLAN Symposium on
  Principles and Practice of Parallel Programming}}
  \emph{(\bibinfo{series}{PPoPP '03})}. \bibinfo{publisher}{Association for
  Computing Machinery}, \bibinfo{address}{San Diego, California, USA},
  \bibinfo{pages}{167–178}.
\newblock
\showISBNx{1581135882}
\urldef\tempurl%
\url{https://doi.org/10.1145/781498.781528}
\showDOI{\tempurl}


\bibitem[\protect\citeauthoryear{Ogles, Aldous, and Mercer}{Ogles
  et~al\mbox{.}}{2019}]%
        {8894270}
\bibfield{author}{\bibinfo{person}{Benjamin Ogles}, \bibinfo{person}{Peter
  Aldous}, {and} \bibinfo{person}{Eric Mercer}.}
  \bibinfo{year}{2019}\natexlab{}.
\newblock \showarticletitle{Proving Data Race Freedom in Task Parallel Programs
  Using a Weaker Partial Order}. In \bibinfo{booktitle}{\emph{2019 Formal
  Methods in Computer Aided Design, {FMCAD} 2019, San Jose, CA, USA, October
  22-25, 2019}}, \bibfield{editor}{\bibinfo{person}{Clark~W. Barrett} {and}
  \bibinfo{person}{Jin Yang}} (Eds.). \bibinfo{publisher}{{IEEE}},
  \bibinfo{address}{San Jose, CA, USA}, \bibinfo{pages}{55--63}.
\newblock
\urldef\tempurl%
\url{https://doi.org/10.23919/FMCAD.2019.8894270}
\showDOI{\tempurl}


\bibitem[\protect\citeauthoryear{Pavlogiannis}{Pavlogiannis}{2020}]%
        {Andreas:Pavlogiannis:popl20}
\bibfield{author}{\bibinfo{person}{Andreas Pavlogiannis}.}
  \bibinfo{year}{2020}\natexlab{}.
\newblock \showarticletitle{Fast, sound, and effectively complete dynamic race
  prediction}.
\newblock \bibinfo{journal}{\emph{Proc. {ACM} Program. Lang.}}
  \bibinfo{volume}{4}, \bibinfo{number}{{POPL}} (\bibinfo{year}{2020}),
  \bibinfo{pages}{17:1--17:29}.
\newblock
\urldef\tempurl%
\url{https://doi.org/10.1145/3371085}
\showDOI{\tempurl}


\bibitem[\protect\citeauthoryear{Roemer, Gen{\c{c}}, and Bond}{Roemer
  et~al\mbox{.}}{2018}]%
        {Roemer:2018:HUS:3296979.3192385}
\bibfield{author}{\bibinfo{person}{Jake Roemer}, \bibinfo{person}{Kaan
  Gen{\c{c}}}, {and} \bibinfo{person}{Michael~D. Bond}.}
  \bibinfo{year}{2018}\natexlab{}.
\newblock \showarticletitle{High-coverage, unbounded sound predictive race
  detection}. In \bibinfo{booktitle}{\emph{Proceedings of the 39th {ACM}
  {SIGPLAN} Conference on Programming Language Design and Implementation,
  {PLDI} 2018, Philadelphia, PA, USA, June 18-22, 2018}},
  \bibfield{editor}{\bibinfo{person}{Jeffrey~S. Foster} {and}
  \bibinfo{person}{Dan Grossman}} (Eds.). \bibinfo{publisher}{{ACM}},
  \bibinfo{pages}{374--389}.
\newblock
\urldef\tempurl%
\url{https://doi.org/10.1145/3192366.3192385}
\showDOI{\tempurl}


\bibitem[\protect\citeauthoryear{Roemer, Gen\c{c}, and Bond}{Roemer
  et~al\mbox{.}}{2020}]%
        {10.1145/3385412.3385993}
\bibfield{author}{\bibinfo{person}{Jake Roemer}, \bibinfo{person}{Kaan
  Gen\c{c}}, {and} \bibinfo{person}{Michael~D. Bond}.}
  \bibinfo{year}{2020}\natexlab{}.
\newblock \showarticletitle{SmartTrack: Efficient Predictive Race Detection}.
  In \bibinfo{booktitle}{\emph{Proceedings of the 41st ACM SIGPLAN Conference
  on Programming Language Design and Implementation}} (London, UK)
  \emph{(\bibinfo{series}{PLDI 2020})}. \bibinfo{publisher}{Association for
  Computing Machinery}, \bibinfo{address}{New York, NY, USA},
  \bibinfo{pages}{747–762}.
\newblock
\showISBNx{9781450376136}
\urldef\tempurl%
\url{https://doi.org/10.1145/3385412.3385993}
\showDOI{\tempurl}


\bibitem[\protect\citeauthoryear{Said, Wang, Yang, and Sakallah}{Said
  et~al\mbox{.}}{2011}]%
        {Said:2011:GDR:1986308.1986334}
\bibfield{author}{\bibinfo{person}{Mahmoud Said}, \bibinfo{person}{Chao Wang},
  \bibinfo{person}{Zijiang Yang}, {and} \bibinfo{person}{Karem Sakallah}.}
  \bibinfo{year}{2011}\natexlab{}.
\newblock \showarticletitle{Generating Data Race Witnesses by an SMT-based
  Analysis}. In \bibinfo{booktitle}{\emph{Proc.\ of NFM'11}}.
  \bibinfo{publisher}{Springer}, \bibinfo{pages}{313--327}.
\newblock


\bibitem[\protect\citeauthoryear{Savage, Burrows, Nelson, Sobalvarro, and
  Anderson}{Savage et~al\mbox{.}}{1997}]%
        {Savage:1997:EDD:265924.265927}
\bibfield{author}{\bibinfo{person}{Stefan Savage}, \bibinfo{person}{Michael
  Burrows}, \bibinfo{person}{Greg Nelson}, \bibinfo{person}{Patrick
  Sobalvarro}, {and} \bibinfo{person}{Thomas Anderson}.}
  \bibinfo{year}{1997}\natexlab{}.
\newblock \showarticletitle{Eraser: A Dynamic Data Race Detector for
  Multithreaded Programs}.
\newblock \bibinfo{journal}{\emph{ACM Trans. Comput. Syst.}}
  \bibinfo{volume}{15}, \bibinfo{number}{4} (\bibinfo{date}{Nov.}
  \bibinfo{year}{1997}), \bibinfo{pages}{391--411}.
\newblock
\showISSN{0734-2071}
\urldef\tempurl%
\url{https://doi.org/10.1145/265924.265927}
\showDOI{\tempurl}


\bibitem[\protect\citeauthoryear{Sen, Ro{\c{s}}u, and Agha}{Sen
  et~al\mbox{.}}{2005}]%
        {sen2005detecting}
\bibfield{author}{\bibinfo{person}{Koushik Sen}, \bibinfo{person}{Grigore
  Ro{\c{s}}u}, {and} \bibinfo{person}{Gul Agha}.}
  \bibinfo{year}{2005}\natexlab{}.
\newblock \showarticletitle{Detecting Errors in Multithreaded Programs by
  Generalized Predictive Analysis of Executions}. In
  \bibinfo{booktitle}{\emph{International Conference on Formal Methods for Open
  Object-Based Distributed Systems}}. \bibinfo{publisher}{Springer},
  \bibinfo{pages}{211--226}.
\newblock


\bibitem[\protect\citeauthoryear{Serbanuta, Chen, and Ro\c{s}u}{Serbanuta
  et~al\mbox{.}}{2012}]%
        {DBLP:conf/rv/SerbanutaCR12}
\bibfield{author}{\bibinfo{person}{Traian{-}Florin Serbanuta},
  \bibinfo{person}{Feng Chen}, {and} \bibinfo{person}{Grigore Ro\c{s}u}.}
  \bibinfo{year}{2012}\natexlab{}.
\newblock \showarticletitle{Maximal Causal Models for Sequentially Consistent
  Systems}. In \bibinfo{booktitle}{\emph{Poc.\ of RV'12}}
  \emph{(\bibinfo{series}{LNCS})}, Vol.~\bibinfo{volume}{7687}.
  \bibinfo{publisher}{Springer}, \bibinfo{pages}{136--150}.
\newblock
\urldef\tempurl%
\url{https://doi.org/10.1007/978-3-642-35632-2_16}
\showDOI{\tempurl}


\bibitem[\protect\citeauthoryear{Serebryany and Iskhodzhanov}{Serebryany and
  Iskhodzhanov}{2009}]%
        {serebryany2009threadsanitizer}
\bibfield{author}{\bibinfo{person}{Konstantin Serebryany} {and}
  \bibinfo{person}{Timur Iskhodzhanov}.} \bibinfo{year}{2009}\natexlab{}.
\newblock \showarticletitle{{ThreadSanitizer}: Data Race Detection in
  Practice}. In \bibinfo{booktitle}{\emph{Proc.\ of WBIA ’09}}.
  \bibinfo{publisher}{ACM}, \bibinfo{address}{New York, NY, USA},
  \bibinfo{pages}{62--71}.
\newblock
\urldef\tempurl%
\url{https://doi.org/10.1145/1791194.1791203}
\showDOI{\tempurl}


\bibitem[\protect\citeauthoryear{Smaragdakis, Evans, Sadowski, Yi, and
  Flanagan}{Smaragdakis et~al\mbox{.}}{2012}]%
        {Smaragdakis:2012:SPR:2103621.2103702}
\bibfield{author}{\bibinfo{person}{Yannis Smaragdakis}, \bibinfo{person}{Jacob
  Evans}, \bibinfo{person}{Caitlin Sadowski}, \bibinfo{person}{Jaeheon Yi},
  {and} \bibinfo{person}{Cormac Flanagan}.} \bibinfo{year}{2012}\natexlab{}.
\newblock \showarticletitle{Sound Predictive Race Detection in Polynomial
  Time}.
\newblock \bibinfo{journal}{\emph{SIGPLAN Not.}} \bibinfo{volume}{47},
  \bibinfo{number}{1} (\bibinfo{date}{Jan.} \bibinfo{year}{2012}),
  \bibinfo{pages}{387--400}.
\newblock
\showISSN{0362-1340}
\urldef\tempurl%
\url{https://doi.org/10.1145/2103656.2103702}
\showDOI{\tempurl}


\bibitem[\protect\citeauthoryear{Smith, Bull, and Obdrizalek}{Smith
  et~al\mbox{.}}{2001}]%
        {smith2001parallel}
\bibfield{author}{\bibinfo{person}{Lorna~A Smith}, \bibinfo{person}{J~Mark
  Bull}, {and} \bibinfo{person}{J Obdrizalek}.}
  \bibinfo{year}{2001}\natexlab{}.
\newblock \showarticletitle{A Parallel {Java} Grande Benchmark Suite}. In
  \bibinfo{booktitle}{\emph{Proc.\ of SC'01}}. \bibinfo{publisher}{IEEE},
  \bibinfo{pages}{8--8}.
\newblock
\urldef\tempurl%
\url{https://doi.org/10.1145/582034.582042}
\showDOI{\tempurl}


\bibitem[\protect\citeauthoryear{Stadtm{\"{u}}ller}{Stadtm{\"{u}}ller}{2021a}]%
        {DBLP:phd/dnb/Stadtmuller21}
\bibfield{author}{\bibinfo{person}{Kai Stadtm{\"{u}}ller}.}
  \bibinfo{year}{2021}\natexlab{a}.
\newblock \emph{\bibinfo{title}{Efficient, near complete and often sound data
  race prediction}}.
\newblock \bibinfo{thesistype}{Ph.D. Dissertation}. \bibinfo{school}{University
  of Freiburg, Freiburg im Breisgau, Germany}.
\newblock
\urldef\tempurl%
\url{https://freidok.uni-freiburg.de/data/222723}
\showURL{%
\tempurl}


\bibitem[\protect\citeauthoryear{Stadtm{\"{u}}ller}{Stadtm{\"{u}}ller}{2021b}]%
        {speedy-go}
\bibfield{author}{\bibinfo{person}{Kai Stadtm{\"{u}}ller}.}
  \bibinfo{year}{2021}\natexlab{b}.
\newblock \bibinfo{title}{{SpeedyGo}}.
\newblock \bibinfo{howpublished}{https://github.com/KaiSta/SpeedyGo}.
\newblock


\bibitem[\protect\citeauthoryear{Sulzmann and Stadtm\"{u}ller}{Sulzmann and
  Stadtm\"{u}ller}{2020}]%
        {10.1145/3426182.3426185}
\bibfield{author}{\bibinfo{person}{Martin Sulzmann} {and} \bibinfo{person}{Kai
  Stadtm\"{u}ller}.} \bibinfo{year}{2020}\natexlab{}.
\newblock \showarticletitle{Efficient, near Complete, and Often Sound Hybrid
  Dynamic Data Race Prediction}. In \bibinfo{booktitle}{\emph{Proceedings of
  the 17th International Conference on Managed Programming Languages and
  Runtimes}} \emph{(\bibinfo{series}{MPLR 2020})}.
  \bibinfo{publisher}{Association for Computing Machinery},
  \bibinfo{address}{Virtual, UK}, \bibinfo{pages}{30–51}.
\newblock
\showISBNx{9781450388535}
\urldef\tempurl%
\url{https://doi.org/10.1145/3426182.3426185}
\showDOI{\tempurl}


\bibitem[\protect\citeauthoryear{ThreadSanitizer}{ThreadSanitizer}{2020}]%
        {thread-sanitizer-v2}
ThreadSanitizer \bibinfo{year}{2020}\natexlab{}.
\newblock \bibinfo{title}{ThreadSanitizer}.
\newblock \bibinfo{howpublished}{https://github.com/google/sanitizers}.
\newblock


\bibitem[\protect\citeauthoryear{Xie, Xue, and Zhang}{Xie
  et~al\mbox{.}}{2013}]%
        {xie2013acculock}
\bibfield{author}{\bibinfo{person}{Xinwei Xie}, \bibinfo{person}{Jingling Xue},
  {and} \bibinfo{person}{Jie Zhang}.} \bibinfo{year}{2013}\natexlab{}.
\newblock \showarticletitle{{Acculock}: Accurate and Efficient Detection of
  Data Races}.
\newblock \bibinfo{journal}{\emph{Software: Practice and Experience}}
  \bibinfo{volume}{43}, \bibinfo{number}{5} (\bibinfo{year}{2013}),
  \bibinfo{pages}{543--576}.
\newblock
\urldef\tempurl%
\url{https://doi.org/10.1109/CGO.2011.5764688}
\showDOI{\tempurl}


\bibitem[\protect\citeauthoryear{Yu and Bae}{Yu and Bae}{2016}]%
        {yu2016simplelock+}
\bibfield{author}{\bibinfo{person}{Misun Yu} {and} \bibinfo{person}{Doo-Hwan
  Bae}.} \bibinfo{year}{2016}\natexlab{}.
\newblock \showarticletitle{SimpleLock+: fast and accurate hybrid data race
  detection}.
\newblock \bibinfo{journal}{\emph{Comput. J.}} \bibinfo{volume}{59},
  \bibinfo{number}{6} (\bibinfo{year}{2016}), \bibinfo{pages}{793--809}.
\newblock
\urldef\tempurl%
\url{https://doi.org/10.1109/PDCAT.2013.15}
\showDOI{\tempurl}


\end{thebibliography}

\pagebreak

\appendix

\section{Weak Causal Precedence (WCP)}
\label{sec:wcp}

We repeat the definition of the
weak-causally precedes (WCP) relation for a well-formed trace $T$.
We point a weak point of the WCP soundness property in case of cross-thread critical sections.

\begin{definition}[Thread Order]
  The \emph{thread-order relation} $\ltTO{T}{}{}$
  is defined for $\evtAA, \evtBB \in T$ by
  $\ltTO{T}{\evtAA}{\evtBB}$
  if
  $\thd{\evtAA} = \thd{\evtBB}$
  and $\ltTrace{T}{e}{f}$.
\end{definition}

\begin{definition}[Critical Sections Happens-Before]
  The \emph{critical sections happens-before relation} $\ltHBCS{T}{}{}$
  is the smallest strict partial order on $T$ that satisfies
  the following two rules:
\begin{description}
  \item[\nf{(HB-1)}] $\ltTO{T}{}{} \subseteq \ltHBCS{T}{}{}$.
  \item[\nf{(HB-2)}] $\ltHBCS{T}{r}{a}$ if $r = (t, \unlockE{x}) \in T$ and $a = (s, \lockE{x}) \in T$
  and $\ltTrace Tra$.
\end{description}
\end{definition}
The (critical sections) happens-before relation does not consider fork-join dependencies.
Such dependencies imply WCP relations as stated by the following definition.

\begin{definition}[Weak-Causally Precedes (WCP)~\citep{Kini:2017:DRP:3140587.3062374}]
\label{def:wcp}
  The \emph{weak-causally precedes relation}
  $\ltWCP{T}{}{}$ is the smallest binary relation on events in $T$ that satisfies
  the following rules:

  \begin{description}
  \item[\nf{(WCP-1)}]
    $\ltWCP{T}{r_1}{\evtBB}$ if
      for any two matching acquire-release pairs $(a_1,r_1)$ and $(a_2,r_2)$
      for some lock $x$ where $\ltTrace{T}{a_1}{a_2}$,
     there is an event $\evtAA$ such that
    $\conflict\evtAA\evtBB$,
    $\evtAA \in \StdCSect{T}{\AcqRelPair{a_1}{r_1}}{x}$, and
    $\evtBB \in \StdCSect{T}{\AcqRelPair{a_2}{r_2}}{x}$.

   \item[\nf{(WCP-2)}]
     $\ltWCP{T}{r_1}{r_2}$ if
      for any two matching acquire-release pairs $(a_1,r_1)$ and $(a_2,r_2)$
      for some lock $x$, we have that
      $\ltWCP{T}{a_1}{r_2}$.

    \item[\nf{(WCP-3)}]
      $\evtAA \ltWCPSym{T} \evtCC$ if
       for some event $\evtBB$ we either have that
       $\evtAA \ltWCPSym{T} \evtBB \ltHBCSSym{T} \evtCC$, or
      $\evtAA \ltHBCSSym{T} \evtBB \ltWCPSym{T} \evtCC$,
    then

  \item[\nf{(WCP-4)}]
    $\ltWCP{T}{\evtAA}{\evtBB}$ if
     $\evtAA = (t, \forkE{s})$ and $s = \thd{\evtBB}$.

  \item[\nf{(WCP-5)}]
    $\ltWCP{T}{\evtAA}{\evtBB}$ if
     $\evtBB = (t, \joinE{s})$ and $s = \thd{\evtAA}$.
  \end{description}

  We define $\lteqWCP{T}{}{} = (\ltWCP{T}{}{} \cup \ltTO{T}{}{})$.

Two conflicting events $\conflict\evtAA\evtBB$
are in a \emph{WCP-race} if neither $\lteqWCP{T}{\evtAA}{\evtBB}$, nor $\lteqWCP{T}{\evtBB}{\evtCC}$.
\end{definition}

Our formulation of WCP differs from the one in the literature~\cite{Kini:2017:DRP:3140587.3062374}:
\begin{itemize}
  \item Rule (WCP-2) is formulated differently from rule (b) in definition 3 in~\cite{Kini:2017:DRP:3140587.3062374}.
        Both formulations are equivalent as shown in appendix~\ref{sec:wcp-orig}.
  \item Fork and join events are not covered in~\cite{Kini:2017:DRP:3140587.3062374}.
        They are supported by the WCP implementation~\cite{rapid-mathur}\footnote{Personal communication
          with one of the implementors.} and behave like write-read dependencies protected by a common lock.
        As we use fork/join in examples, we explitely add the corresponding rules (WCP-4) and (WCP-5).
\end{itemize}
Besides these differences, the above faithfully captures WCP.

Before restating the WCP soundness property, we repeat
the definition of a predictable deadlock that is commonly used in the literature.

\begin{definition}
  \label{def:pred-deadlock}
  We say that a trace $T$ exhibits a \emph{predictable deadlock}
  if there exists a correctly reordered prefix $T'$ of $T$,
  distinct threads  $\THD{1},\dots ,\THD{n}$,
  distinct locks $x_1, \dots, x_n$,
  acquire events $\ACQ{1},\dots ,\ACQ{n} \in T$ and
  events $e_1, ..., e_n \in T$
  where $n>1$, $\thd{e_i} = \THD{i}$ and $\ACQ{i} = (\THD{i}, \unlockE{x_i})$ such the following condition hold:

\begin{description}
\item[PD-1:] For each $i\in\{i,...,n\}$ we have that
              $\proj{\THD{i}}{T'} = [...,e_i]$ and
              $\proj{\THD{i}}{T} = \proj{\THD{i}}{T'} \pp\ [\ACQ{i},\dots ]$.

  \item[PD-2:] For each $i\in\{i,...,n\}$ we have that
                $x_i \in \StdLHeld{T}{e_j}$ for some $j \not= i$.
  \end{description}
\end{definition}
The above says that for each thread $\THD{i}$, the trace $T'$ cannot be extended to include
$\ACQ{i}$ because the associated lock $x_i$ is hold by another thread.

The WCP soundness result assumes that traces are well-nested.

\begin{definition}
  A well-formed trace $T$ is \emph{well-nested}
  if for any two matching acquire-release pairs $(a_1,r_1)$ and $(a_2,r_2)$
  and lock $x$ we do not have that $r_2 \in \StdCSect{T}{\AcqRelPair{a_1}{r_1}}{x}$.
\end{definition}
The definition states that if a thread acquires lock $y$ followed by acquiring lock $x$,
the matching release of $y$ cannot appear before the matching release of $x$.

\begin{theorem}[Soundness of WCP, theorem 1 in~\citep{Kini:2017:DRP:3140587.3062374}]
WCP is weakly sound, i.e.,
given any well-formed and well-nested trace $T$, if $T$ exhibits a WCP-race then $T$ exhibits a
predictable race or a predictable deadlock.
\end{theorem}

WCP and the prior CP relation are major accomplishments.
Initially, the WCP authors tried to fix the soundness proof for CP~\cite{Smaragdakis:2012:SPR:2103621.2103702}.
This turned out to be unsuccessful which then lead to WCP, a \emph{weaker} CP version.
Multiple years of effort went into CP as well as WCP.
The soundness proof of WCP owes much to the efforts put into CP and its soundness proof.
Common to both is that the WCP and CP relation are agnostic about cross-thread critical sections.

\subsection{WCP Soundness Proof Structure}

The first part of the WCP soundness proof is along the lines of the CP soundness proof.
We can show that the reasoning in the proof can be extended to include cross-thread critical sections.
See appendix~\ref{sec:wcp-sound-cross-thread}.

The second part of the WCP soundness proof goes beyond the CP proof by reasoning about deadlock patterns and chains.
This shall lead to a predictable deadlock.
It is this part where we seen an issue as illustrated by the example in figure~\ref{fig:exWCP_9d}.

We examine both parts and stay close to the naming conventions in the CP/WCP soundness proof.

\mbox{}
\\
\noindent
{\bf Proof outline.}
Given some well-formed trace $T$ that exhibits a WCP-race.
Pick the first WCP-race $(e_1, e_2)$.
To goal is to show that there is either a predictable race or a predictable deadlock.

\mbox{}
\\
\noindent
{\bf There is a predictable race.}
If $e_1$ and $e_2$ are unordered under the standard happens-before relation, then there must be a reordering
where $e_1$ and $e_2$ appear right next to each other. So, there is a predictable race and we are done for this case.

\mbox{}
\\
\noindent
{\bf There is no predictable race.}
Otherwise, $\ltHBNoT{e_1}{e_2}$ where we assume that $e_1$ appears before $e_2$ in the trace.
The goal for this case is to show that no predictable race exists.
The proof for this part relies on picking an appropriate reordering $T'$ of $T$.
By exploiting the assumption that $(e_1,e_2)$ are in a WCP-race but $\ltHBNoT{e_1}{e_2}$
this then leads to a contradiction.
Parts of this reasoning only consider the case of standard critical sections.
In appendix~\ref{sec:wcp-sound-cross-thread} we show that the reasoning can be extended
to include the case of cross-thread critical sections as well.

\mbox{}
\\
\noindent
{\bf There is a predictable deadlock.}
The rest of WCP soundness proof establishes properties about deadlock patterns and chains.
\begin{quote}
  "What remains to be done is to prove that deadlock chains
result in predictable deadlocks. In order to do so we introduce
intermediate structures called deadlock patterns."
\end{quote}

Based on our understanding, the definitions of deadlock patterns/chains
only take into accout the deadlocked threads.
As shown by the example in figure~\ref{fig:exWCP_9d}, this might be insufficient.

Initially, we tried to disprove the WCP soundness property
via some cross-thread critical section example where
there is a WCP-race that is not predictable and there is a ``deadlocking'' situation
that goes beyond the standard notion of a predictable deadlock.
We failed many times.

It seems that whenever there is a
``deadlocking'' situation
that goes beyond the standard notion of a predictable deadlock
there cannot be WCP-race.
For an example, see figure~\ref{fig:exWCP_9d}.
The exact details of this claim, yet need to be worked out.

\section{Comparison to Original WCP Definition}
\label{sec:wcp-orig}

The original formulatin of WCP-2 reads as follows.

  \begin{description}
   \item[\nf{(WCP-2')}]
     $\ltWCP{T}{r_1}{r_2}$ if
      for any two matching acquire-release pairs $(a_1,r_1)$ and $(a_2,r_2)$
      for some lock $x$, we have that
      for some events $\evtAA \in \StdCSect{T}{\AcqRelPair{a_1}{r_1}}{x}$ or $\evtAA = a_1$, and
      $\evtBB \in \StdCSect{T}{\AcqRelPair{a_2}{r_2}}{x}$ or $\evtBB = r_2$ we have that
      $\ltWCP{T}{\evtAA}{\evtBB}$.
  \end{description}

Formulations WCP-2 and WCP-2' are equivalent.
That WCP-2' subsumes WCP-2 follows immediately.
For the other direction.
Consider the case that $\evtAA \in \StdCSect{T}{\AcqRelPair{a_1}{r_1}}{x}$, $\evtBB \in \StdCSect{T}{\AcqRelPair{a_2}{r_2}}{x}$  and $\ltWCP{T}{\evtAA}{\evtBB}$.

\begin{enumerate}
 \item Because we consider standard critical sections, we find that $\ltHB{T}{a_1}{\evtAA}$
and $\ltHB{T}{\evtBB}{r_2}$.

 \item Via left and right composition with WCP, we find that $\ltWCP{T}{a_1}{r_2}$.

 \item Via rule (WCP-2), we can establish $\ltWCP{T}{r_1}{r_2}$ and we are done.
\end{enumerate}

\section{WCP Soundness Revisted}
\label{sec:wcp-sound-cross-thread}

We consider parts of the WCP soundness proof.

The general assumption is as follows:
  Let $T$ be a well-formed trace with a WCP race $(e_1, e_2)$
  where $e_1$ appears before $e_2$ in the trace.
  Let $T'$ be a correct reordering such that
  (a) the distance between $e_1$ and $e_2$ is minimal, and
  (b) the distance from $e_2$ to every acquire that
  encloses $e_1$ is minimal (from innermost to outermost acquires).

We repeat two essential lemmas.

\begin{lemma}[Lemma A.1 in WCP technical report~\cite{DBLP:journals/corr/KiniM017}]
  For all events e such that $e_1 \ltTraceSym{T'} e \ltTraceSym{T'} e_2$ we have that
\begin{enumerate}
\item $e_1 \ltHBSym{T'} e \ltHBSym{T'} e_2$.
\item $e_1 \not \ltWCPSym{T'} e$ and $e \not \ltWCPSym{T'} e_2$.
\end{enumerate}
\end{lemma}

\begin{lemma}[Lemma A.2 in WCP technical report~\cite{DBLP:journals/corr/KiniM017}]
  Let $a_1$ be an acquire event such that $e_1 \in \StandardCS{T'}{a_1}$.
  For all events $e$ such that   $\posP{T'}{a_1} < \posP{T'}{e} < \posP{T'}{e_1}$
  we have that $\ltHB{T'}{a_1}{e}$ and $\ltHB{T'}{e}{e_2}$.
\end{lemma}

The above corresponds to lemma 2 in~\cite{Smaragdakis:2012:SPR:2103621.2103702}.
There is no proof for lemma A.1 and lemma A.2 in the WCP technical report version because
they are analogous to lemma 1 and 2 in~\cite{Smaragdakis:2012:SPR:2103621.2103702}.

The proof of lemma 2 in~\cite{Smaragdakis:2012:SPR:2103621.2103702}
assumes standard critical section.
We show that considering standard critical sections only is valid as we
otherwise reach a contradiction to the assumption that $e_1$ and $e_2$ are in a WCP race.
Here are the details.

To verify $\ltHB{T'}{a_1}{e}$, the proof of lemma 2 in~\cite{Smaragdakis:2012:SPR:2103621.2103702}
assumes the contrary.
That is, $a_1 \not \ltHBSym{T'} e$.

\begin{enumerate}
\item Consider $E = \{ e' \mid \ltHB{T'}{e'}{e} \wedge a_1 \ltTraceSym{T'} e' \ltTraceSym{T'} e \}$.
  This is the set of all events that occur between (and inclusive of) $a_1$ and $e$ are in a happens-before
  relation with $e$.

\item The argument is now that the set $E$ cannot contain any events from thread $t_1$
  where $t_1 = \thd{e_1}$, or else $\ltHB{T'}{a_1}{e}$.

\item This reasoning can only go through in case of standard critical sections.
  In this case, $a_1$ is in thread $t_1$ as well.

\item Suppose, there is some $e'$ where $t_1 = \thd{e'}$ and $\ltHB{T'}{e'}{e}$
  and $a_1 \ltTraceSym{T'} e'$.

\item Due to thread-order, we find that $\ltHB{T'}{a_1}{e'}$
      and thus $\ltHB{T'}{a_1}{e}$ which contradicts the assumption.
\end{enumerate}

In general, $a_1$ could belong to a cross-thread critical section.
We show that it is safe to eliminate such cases here.

Case $e_1 \in  \IntraCSect{T}{\AcqRelPair{a_1}{r_1}}{x}$
      where $r_1$ is the matching release for $a_1$ and $x$ is the name of the lock.
      We assume that $a_1, r_1$ are not in the same thread as $e_1$.

\begin{enumerate}
\item We have that $r_1$ \emph{must} happen after $e_1$.

\item This implies that there must write-read dependency in between $e_1$ and $r_1$
      such that $\ltMustHB{T'}{e_1}{r_1}$.

    \item Hence, we must find a write $e_w$ and a read $e_r$ such that
      $e_1 \ltTraceSym{T'} e_w \ltTraceSym{T'} e_r \ltTraceSym{T'} r_1$
      where $e_w$ is the last write for $e_r$.

\item This write-read dependency must be protected by a common lock~$y$.
      Otherwise, we obtain a contradiction to the assumption that $e_1$ and $e_2$
      are the first race.

\item We can further assume that
      (1) $t_1 = \thd{e_w}$ and $r_1$ and $e_r$ are in some
          thread $t \not = t_1$, and
      (2) $e_w \in \StandardCS{T'}{y}$ and $e_r \in \StandardCS{T'}{y}$.
         That is, $e_w$ and $e_r$ are part of a standard critical section.

\item In general, there could be a sequence of write-read dependencies
      where some of them are protected by  cross-thread critical sections.
      However, we must find a write-read dependency that is only protected by standard critical sections.

\item Via the above and WCP rule (WCP-1) we conclude that $\ltWCP{T'}{e_w}{e_r}$.

\item Via lemma A.1 we can conclude that $\ltHB{T'}{e_1}{e_w}$ and $\ltHB{T'}{e_r}{e_2}$.

\item Via WCP rule (WCP-3) we can conclude that $\ltWCP{T'}{e_1}{e_2}$.

\item This contradicts the assumption that $e_1$ and $e_2$ are in a WCP race.

\end{enumerate}

Hence, it is safe to assume that $e_1$ is part of a standard critical section.

\newcommand{\pwrE}{PWR-E}

\newcommand{\pwrsymbole}{\textit{pwr-e}}
\newcommand{\ltPWRE}[3]{#2 <_{\pwrsymbole}^{#1} #3}

\newcommand{\ltPWRENoT}[2]{#1 <_{\pwrsymbole} #2}
\newcommand{\CTTLHeldNoT}[1]{\LHSym{\intrathreadthread}{(#1)}}

\section{Full Enhancement of PWR}
\label{sec:further-pwr}

We consider the full enhancement of PWR with cross-thread critical sections.

\begin{definition}[\pwrE]
  \label{def:pwre-relation}
  For a well-formed trace $T$,
  $\ltPWRE{T}{}{}$
  is the smallest strict partial order that
  satisfies the following rules:

  \begin{description}
  \item[\nf{(PWRE-1)}:]
        $\ltPWRE{T}{\evtAA}{\evtBB}$ if
       $\ltTrace{T}{\evtAA}{\evtBB}$ and $\thd{\evtAA} = \thd{\evtBB}$.

  \item[\nf{(PWRE-2)}:]
      $\ltPWRE{T}{\evtBB}{\evtAA}$ if
        $\evtAA = (t,\readE{a})$ and
        $\evtBB$ is the last write for $\evtAA$ w.r.t.~$T$.
  \item[\nf{(PWRE-3)}:]
    $\ltPWR{T}{r_1}{f}$ if
    for any two matching acquire-release pairs $(a_1,r_1)$ and $(a_2,r_2)$
    and for some lock $x$ where $\ltTrace{T}{a_1}{a_2}$
    there is some event $e$ such that
    $e \in \IntraCSect{T}{\AcqRelPair{a_1}{r_1}}{x}$ and
    $f \in \IntraCSect{T}{\AcqRelPair{a_2}{r_2}}{x}$.
  \item[\nf{(PWRE-4)}]
    $\ltPWRE{T}{\evtAA}{\evtBB}$ if $\evtAA = (t, \forkE{s})$ and $s = \thd{\evtBB}$,

  \item[\nf{(PWRE-5)}]
    $\ltPWRE{T}{\evtAA}{\evtBB}$ if $\evtBB = (t, \joinE{s})$ and $s = \thd{\evtAA}$.
  \end{description}

Two conflicting events $\conflict\evtAA\evtBB$
are in a \emph{\pwrE-race} if neither $\ltPWRE{T}{\evtAA}{\evtBB}$, nor $\ltPWRE{T}{\evtBB}{\evtCC}$,
and $\CTTLHeldNoT{\evtAA} \cap \CTTLHeldNoT{\evtBB} = \emptyset$.

\end{definition}

The enhancements compared to PWR are as follows:
(E1) In rule (PWRE-3) we employ cross-thread critical sections
and (E2) the \pwrE-race checks employs thread indexed cross-thread lock sets.

\begin{figure}

\bda{|l|l|l|l|l|}
\hline  & \thread{1} & \thread{2} & \thread{3} & \thread{4}\\ \hline
\eventE{1}  & \forkE{\thread{3}}&&&\\
\eventE{2}  & \lockE{x}&&&\\
\eventE{3}  & \forkE{\thread{2}}&&&\\
\eventE{4}  & &\writeE{a}&&\\
\eventE{5}  & \writeE{b}&&&\\
\eventE{6}  & \joinE{\thread{2}}&&&\\
\eventE{7}  & \unlockE{x}&&&\\
\eventE{8}  & &&\lockE{x}&\\
\eventE{9}  & &&\forkE{\thread{4}}&\\
\eventE{10}  & &&&\readE{a}\\
\eventE{11}  & &&\joinE{\thread{4}}&\\
\eventE{12}  & &&\unlockE{x}&\\
\eventE{13}  & &&\writeE{b}&\\

\hline \eda{}

    \caption{Enhancement (E1) by example}
\label{fig:exCS_5b}

\end{figure}

How to integrate (E2) into PWR is described in algorithm~\ref{alg:pwr_guards}.
(E1) we effectively get for free.

Consider the example in figure~\ref{fig:exCS_5b}.
There is no race here for the following reasoning.
\begin{enumerate}
\item $\ltPWRENoT{e_4}{e_{10}}$ via rule (PWRE-2).
\item From this we derive $\ltPWRENoT{e_6}{e_{11}}$.
\item Via rule (PWRE-3) we find that $\ltPWRENoT{e_7}{e_{11}}$.
\item Therefore, $\ltPWRENoT{e_5}{e_{13}}$.
\end{enumerate}

In our implementation, we keep track of critical sections (but limiting their number for efficiency reasons).
When processing the release event $e_{12}$, we check for any prior critical section.
Here, we find $(V_{e_2}, V_{e_7})$, the critical section in thread $t_1$ represented as pair of
vector clocks where $V_{e_2}$ corresponds to the acquire and $V_{e_7}$ corresponds to the release.
Due to the fork/join and write-read dependencies, we find that $V_{e_2} < V_{e_{12}}$
and therefore thread $t_3$ synchronizes with $V_{e_7}$.
Thus, we find that $V_{e_5} < V_{e_{13}}$.

As discussed, our implementation underapproximates both enhancements.
\begin{itemize}
   \item The implementation may not recognize a cross-thread critical section and fails to apply rule (PWRE-3).
   \item The implementation may not recognize that a lock is in the set $\CTTLHeldNoT{}$.
\end{itemize}


\end{document}